\documentclass[12pt]{article}
\usepackage[margin=1in]{geometry}
\usepackage{paralist}
\RequirePackage{amsthm,amsmath,amsfonts,amssymb}
\RequirePackage{algorithm,algorithmic}
\usepackage{thmtools}
\RequirePackage{xcolor}
\RequirePackage{bbm}
\RequirePackage{textcomp}
\RequirePackage{gensymb}
\RequirePackage{centernot}
\RequirePackage{bm}
\RequirePackage{enumerate}
\RequirePackage{enumitem}
\RequirePackage{graphicx}
\RequirePackage{multirow}
\RequirePackage{booktabs}
\RequirePackage{mathtools}
\RequirePackage[authoryear,round]{natbib}

\theoremstyle{plain}
\newtheorem{theorem}{Theorem}
\newtheorem*{theorem*}{Theorem}

\newtheorem*{claim*}{Claim}

\newtheorem*{lemma*}{Lemma}
\newtheorem{proposition}{Proposition}
\newtheorem*{proposition*}{Proposition}

\theoremstyle{remark}

\newtheorem{assumption}{Assumption}

\renewcommand\thmcontinues[1]{Continued}

\usepackage{crossreftools}
\makeatletter
\newcommand{\optionaldesc}[2]{%
  \phantomsection
  #1\protected@edef\@currentlabel{#1}\label{#2}%
}
\makeatother

\usepackage{hyperref}
\usepackage{caption} 

\newcommand{\eff}{\textnormal{eff}}
\newcommand{\expit}{\textnormal{expit}}

\newcommand{\rct}{\textnormal{rct}}
\newcommand{\obs}{\textnormal{obs}}

\newcommand{\indic}{\bm{1}}

\newcommand{\cf}{\mathcal{F}}
\newcommand{\ch}{\mathcal{H}}
\newcommand{\ci}{\mathcal{I}}

\newcommand{\cx}{\mathcal{X}}

\newcommand{\aipsw}{\textnormal{aipsw}}

\renewcommand{\leq}{\leqslant}

\renewcommand{\geq}{\geqslant}
\renewcommand{\Pr}{\textnormal{Pr}}

\newcommand{\e}{\mathbb{E}}
\newcommand{\real}{\mathbb{R}}

\newcommand{\var}{\mathrm{Var}}

\newcommand{\collab}{\textnormal{collab}}

\DeclareMathOperator*{\argmin}{arg\,min}

\title{Transporting treatment effects by calibrating large-scale observational outcomes}
\author{Harrison H. Li, Department of Mathematics, Harvey Mudd College}
\date{}

\begin{document}
\maketitle

\begin{abstract}
A high-quality experimental dataset is often much smaller than a corresponding observational dataset. When this holds with possibly biased measurements of the outcome of interest in the latter, we propose an estimation and inference procedure for a transported treatment effect. Our point estimator can be computed as follows. First, we estimate the conditional average treatment effect (CATE) by calibrating a treatment-control contrast estimated using the observational outcomes to the experimental dataset using ordinary least squares (OLS). Then, we compute the sample average of this estimated CATE over the observational dataset. We show that the limiting estimand is a weighted transported average treatment effect even when the OLS calibration is misspecified. Furthermore, our inference for this estimand is asymptotically valid and semiparametrically efficient when the size of the experimental dataset grows more slowly than the size of the observational dataset, regardless of the existence of positivity (overlap) between the two datasets. We illustrate the stable empirical performance of our method under varying degrees of positivity using numerical simulations and a data example using field experiments and satellite-based yield estimates to estimate the average effect of crop rotation on maize (corn) yields over a large area of the Midwestern United States.
\end{abstract}

\section{Introduction}
\label{sec:intro}

Randomized experiments are often conducted on subjects that differ greatly from the population of interest.
This creates a problem of external validity,
underscoring the need for methods to transport treatment effect estimates from the experimental population to a new target population defined by an auxiliary dataset~\citep{degtiar2023review}.

Standard methods for causal transportation are based on two identifying assumptions: mean exchangeability and positivity (also known as overlap).
Mean exchangeability assumes that the conditional means of the unobserved potential outcomes given observed covariates are the same between the experimental and target populations~\citep{dahabreh2020extending},
and positivity requires the support of the covariate distribution in the target population to be contained within that of the experimental population.
In other words, conditional on the covariates for an observation sampled uniformly at random from the superpopulation corresponding to the union of the experimental and observational datasets,
its probability of being from the experimental dataset ---
also known as its sampling propensity ---
needs to be bounded away from zero.
Due to the curse of dimensionality,
violations of positivity are likely when the experimental dataset is small and/or the covariate dimension is even moderate~\citep{d2021overlap}.
This is a highly relevant concern given that
experimentation is often expensive,
greatly limiting the sample size,
and that it is often not feasible to randomize treatment on the specific target population of interest.

As a motivating example,~\citet{kluger2022combining} sought to understand the causal effects of maize-soybean crop rotation on maize (corn) yields across a study region of interest in the Midwestern United States.
The authors compiled
a limited number of multi-year controlled field experiments of crop rotation across the study region,
though unfortunately, their scope did not fully cover the study region.
For instance, there were no available experiments at all in Indiana and Michigan (Fig.~\ref{fig:map}).
To address this,
the authors supplemented the experimental data with much more abundant satellite estimates of crop yields covering the entire study region.
Specifically, they fit a causal forest~\citep{wager2018estimation} using these satellite yield estimates as the outcome to estimate the conditional average treatment effect (CATE) of crop rotation as a function of location, year, and weather and soil covariates. 
Since these CATE estimates may be biased due to unobserved confounders and potential systematic errors in the satellite-based yield estimates,
they were calibrated via linear regression using the limited experimental dataset.

We propose taking a simple average of these calibrated treatment effect estimates over the covariates in the observational dataset to estimate an aggregate transported treatment effect.
This avoids the need to perform inverse propensity weighting and thus remains stable even under severe positivity violations.
While this procedure may seem naive,
we show it is both theoretically principled and empirically performant in the common setting with limited experimental data and plentiful observational data.
Specifically,
when the size of the experimental dataset grows sufficiently slowly relative to the size of the observational dataset,
the procedure enables valid asymptotic inference for a new estimand that remains identified in the absence of positivity.
Under mean exchangeability,
this estimand is precisely equal to the standard transported ATE for the target population when treatment effect heterogeneity estimates in the large observational dataset can be effectively calibrated using the small experimental dataset via a parametric correction.
By contrast,
popular existing proposals for alternative estimands in the face of positivity violations systematically downweight observations in the target population with covariate values that are unlikely (or impossible) to occur in the experimental dataset.

\begin{figure}
\centering
\includegraphics[width=0.8\linewidth]{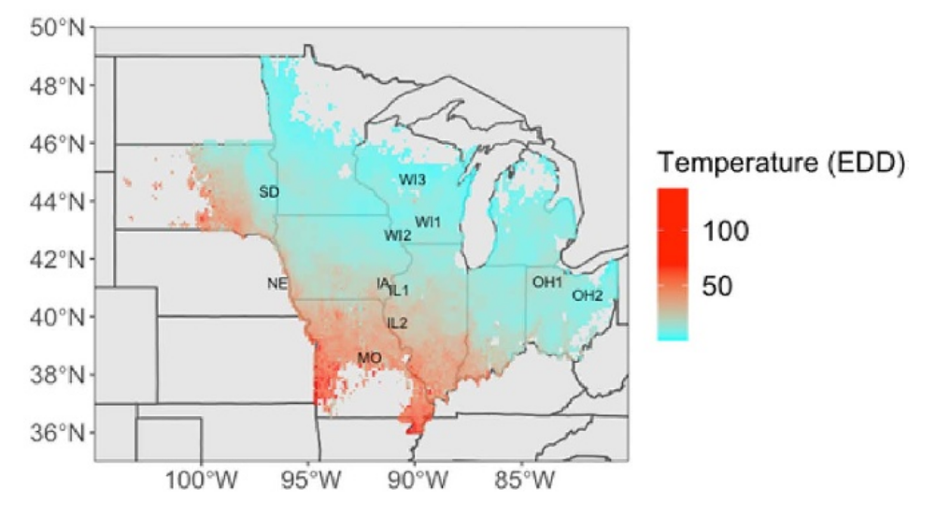}
\caption{A map from Fig. 1 of~\citet{kluger2022combining} showing the 11 locations with available data from crop rotation field experiments (letters), along with colored dots showing the locations and average number of extreme degree days (EDD) of extensive satellite observations spanning the region of interest.}
\label{fig:map}
\end{figure}

\subsection{Related work}
For comprehensive reviews of recent work on causal generalization and transportation as well as related data fusion tasks,
we refer the reader to~\citet{colnet2024causal} and~\citet{rosenman2025methods}.
Here we briefly discuss some related work that addresses our setting with highly unbalanced datasets (i.e., a much smaller experimental dataset compared to the observational datasets) and/or positivity violations.

One important line of work has sought modifications to inverse propensity weighting to make it less sensitive to positivity violations,
without changing the estimand of interest.
Many of these proposals work in the setting of estimating the ATE from a single unconfounded observational dataset (rather than causal transportation),
in which case positivity corresponds to overlapping supports of the covariate distributions of the treated and untreated subjects.
The ideas can be readily adapted to the causal transportation problem;
here we summarize them using the language of causal transportation.
For example,
balancing weights~\citep{li2018balancing, ben2021balancing, josey2022calibration}
and Riesz regression~\citep{chernozhukov2021automatic,chernozhukov2022riesznet,lee2025rieszboost}
have been proposed to make inverse propensity weights more stable.
When outcome data is available in the observational dataset,
as in our setting,
a predictive model can be learned from the external dataset to attempt to explain some of the outcome variance in the main dataset,
thereby increasing the precision of inverse propensity weighting estimators under both design-based~\citep{gagnon2023precise} and superpopulation-based~\citep{huang2023leveraging, liao2025prognostic} frameworks.
Another recent idea (proposed specifically for causal transportation) is to use a ``collaborative" propensity score that is a univariate function of the outcome regression function in place of the actual propensity score which is a function of all covariates~\citep{rudolph_improving_2025}.
This mitigates the extent to which the curse of dimensionality exacerbates positivity violations.

All of the above approaches, however,
are designed to address minor or practical positivity violations rather than population level positivity violations.
Fundamentally,
the average causal effect in the target population is no longer identified when positivity does not hold in the population,
making it is necessary to change the estimand,
as we do in the present work.
Existing proposals for estimands from the ATE literature in the absence of positivity trim and/or downweight observations with extreme propensities either explicitly~\citep{crump2009dealing,li2019addressing,khan2024doubly, huang2025overlap}
or implicitly through modifications to balancing weights~\citep{barnard2025partiallyretargetedbalancingweights}.
In this spirit,~\citet{parikh2025we} propose a decision-tree based algorithm for causal transportation to dynamically learn a subpopulation with sufficiently high sampling propensities (as well as low outcome variance),
and report an ATE on this subpopulation.
While the use of decision trees ensures this subpopulation is interpretable,
the proposed algorithm still requires population-level positivity,
and the resulting inference procedure currently does not account for the complex uncertainty introduced by learning the subpopulation in a data-dependent way.
By contrast,
we rely on the existence of an external observational dataset to target an alternative estimand that is not data-dependent and maintains the scientific interest of equally weighting all covariate observations in the target population when the calibration is correctly specified,
while maintaining interpetability as a weighted average treatment effect with a misspecified calibration.

The problem of transferring information learned from a small dataset to a much larger dataset has also been extensively studied outside causal inference.
Indeed, it is the central problem of semi-supervised learning,
also recently referred to as ``prediction-powered inference."
In that literature,
the small dataset is labeled with both covariates and the outcome of interest while the large dataset is unlabeled,
containing just covariates,
and the goal is to use the unlabeled data to improve estimation of some feature of the outcome distribution,
such as the mean.
There has been significant attention to the asymptotic regime where the ratio of the sample size of the labeled dataset to that of the unlabeled dataset approaches zero,
though mostly under the assumption that the covariate distribution is the same between the labeled and unlabeled datasets~\citep{zhang2019semi, zhang2022high, song2024general, angelopoulos2023prediction, angelopoulos2023ppi++, zrnic2024cross, xu2025unified},
in which case positivity is not a concern.
However,
more recently there has been rising interest in the case where there is unknown covariate shift and ``decaying positivity" as the sample sizes grow~\citep{zhang2023double, su2024solving, testa2025semiparametric}.
While the methods in these works still require some degree of positivity
as they aim to target the original estimand,
they uncover valuable theoretical insights by specifying precise positivity assumptions needed in a non-standard regime for inverse propensity weighted methods to enable valid asymptotic inference.
We find in our simulations and data example, however, that such theoretical guarantees can be delicate in practice.

We also mention some related work to improve the efficiency of causal estimators via data fusion by making structural, low-dimensional parametric assumptions linking multiple datasets~\citep{kallus2018removing, li2024efficient, yang2024data, li2025data,guan2025data, wu2025comparative}.
A key distinction of the present work is that we fully investigate the behavior of our estimator when such assumptions are not satisfied.

The remainder of this paper is structured as follows.
Section~\ref{sec:estimand} mathematically specifies the setting of interest and studies the estimand $\bar{\tau}$ targeted by our proposed estimation and inference procedure,
which is described in detail in Section~\ref{sec:estimation}.
We provide a careful discussion of semiparametric efficiency in Section~\ref{sec:discussion},
showing that our procedure attains the efficiency bound 
for our estimand $\bar{\tau}$ under an asymptotic regime where the size of the observational dataset grows more quickly than the size of the experimental dataset.
We validate the performance of our methodology using a numerical simulation and data example based on the crop rotation dataset used by~\citet{kluger2022combining} in Section~\ref{sec:simulations}.
Finally, Section~\ref{sec:conclusion} concludes.

\section{An alternative estimand}
\label{sec:estimand}
Our formal setup assumes observing experimental samples $\{(D_i,X_i)\}_{i=1}^n$,
that are independent and identically distributed (i.i.d.) alongside i.i.d. observational triples $\{(Y_i,Z_i,X_i)\}_{i=n+1}^{n+N}$ independent of the experimental data.
Here, $X_i \in \cx \subseteq \real^p$ is a vector of observed covariates,
$D_i$ is an experimental observation of the treatment effect,
$Y_i$ is the observational outcome of interest, and
$Z_i$ is a binary treatment indicator in the observational dataset.
In the crop rotation example from Section~\ref{sec:intro},
$D_i$ is the yield in a rotated field minus the yield in a paired adjoining field that was not rotated,
$X_i$ is a vector of location, year, weather, and soil covariates (observed in both the experimental and observational datasets),
$Y_i$ is the satellite-based yield estimate,
and $Z_i$ is an indicator of whether the field underwent crop rotation.
Throughout,
we let $f_{\rct}$ and $f_{\obs}$ denote the densities (with respect to some common dominating measure $\lambda$) for the distributions of the covariates $X_i$ in the experiment and in the observational dataset, respectively.
Subscripts $\rct$ and $\obs$ on probabilities and expectations indicate they are to be taken under the distributions specified by the densities $f_{\rct}$ and $f_{\obs}$, respectively.
We assume throughout that $f_{\obs}$ dominates $f_{\rct}$:

\begin{assumption}
\label{assump:dominance}
$f_{\rct}(x)=0$ for $\lambda$-almost all $x \in \cx$ such that $f_{\obs}(x)=0$.
\end{assumption}

Assuming that treatment is assigned independently of the outcomes $Y_i$ given the covariates $X_i$,
we can extend all of our results to experiments where treatment and control observations are not paired
(i.e., the experimental dataset has observations of $(Y_i,Z_i)$ in place of $D_i$)
 by setting
\[
D_i = \frac{Z_iY_i}{e(X_i)} - \frac{(1-Z_i)Y_i}{1-e(X_i)}
\]
where
\[
e(x) = \Pr_{\rct}(Z=1 \mid X=x)
\]
is the propensity score in the experimental dataset
(typically constant, and known since it is controlled by the experimenter).

We would like to estimate the ATE in the target population defined by the observational dataset.
Under the mean exchangeability assumption that $\e_{\rct}[D \mid X] = \e_{\obs}[D \mid X]$
(note $D$ is not observed in the observational dataset,
so this is untestable),
this ideal estimand is identified as
\[
\tau = \e_{\obs}[\mu(X)]
\]
where
\[
\mu(x) = \e_{\rct}[D \mid X=x]
\]
is the CATE.
The ``transported ATE"
$\tau$ averages the CATE $\mu(x)$ from the experimental dataset
over the covariate distribution of the observational dataset.

The estimand $\tau$, however, is only identified under positivity:

\begin{assumption}[Positivity]
\label{assump:positivity}
$f_{\obs}(x)=0$ for $\lambda$-almost all $x \in \cx$ such that $f_{\rct}(x)=0$.
\end{assumption}
When Assumption~\ref{assump:positivity} does not hold,
for any $x$ where $f_{\obs}(x)>0$ but $f_{\rct}(x)=0$,
$\mu(x)$ is not well-defined,
and so $\tau$ is not well-defined.
To rectify this,
we target an alternative estimand
\begin{equation}
\label{eq:tau_bar}
\bar{\tau} = \e_{\obs}[\bar{\mu}(X)], \quad \bar{\mu}(\cdot) = \argmin_{f \in \cf} \e_{\rct}[(D-f(X))^2] = \argmin_{f \in \cf} \e_{\rct}[(\mu(X)-f(X))^2]
\end{equation}
where for a known basis expansion $\psi:\real \rightarrow \real^p$,
\[
\cf = \{x \mapsto \beta^{\top}\psi(\Delta(x)), \beta \in \real^p \}
\]
is the set of all linear transformations of $x \mapsto \psi(\Delta(x))$
for
\[
\Delta(x) = \e_{\obs}[Y \mid X=x,Z=1] - \e_{\obs}[Y \mid X=x,Z=0].
\]
the ``treatment-control contrast" in the observational dataset.
We do not assume $\Delta=\mu$ due to the possibility of unobserved confounding and/or other potential sources of bias,
such as systematic errors in the satellite estimates of yield for the crop rotation dataset.
Note that every function in $\cf$ has domain that must include all $x$ in the support of $f_{\obs}$,
since $\Delta(x)$ is defined for all such $x$.
Hence, the estimands
$\bar{\mu}(\cdot)$ and $\bar{\tau}$ remain well-defined even when Assumption~\ref{assump:positivity} does not hold.

In lieu of averaging the true CATE $\mu(\cdot)$ over the observational covariate distribution,
the estimand $\bar{\tau}$ averages $\bar{\mu}(\cdot)$,
the best approximation in mean square (with respect to the experimental covariate distribution specified by the density $f_{\rct}$) to $\mu$ within the function class $\cf$.
Concretely, if $\psi(x) = (1,x)^{\top}$,
then the function $\bar{\mu}(\cdot)$ can be interpreted as the limiting prediction function learned by a simple linear regression of $D_i$ onto $\Delta(X_i)$ as the experimental sample size $n$ grows,
and the estimand $\bar{\tau}$ is simply the average of the predictions from this regression over the target population.

It is useful to note that if the components of $\psi(\Delta(X))$ are not linearly dependent w.p.1. in the RCT,
then the first-order conditions for $\bar{\mu}$ imply
\begin{equation}
\label{eq:mu_bar_explicit}
\bar{\mu}(x) = \psi(\Delta(x))^{\top}\bar{\beta}, \qquad \bar{\beta}=\left(\e_{\rct}[\psi(\Delta(X))\psi(\Delta(X))^{\top}]\right)^{-1}\e_{\rct}[\mu(X)\psi(\Delta(X))].
\end{equation}

Proposition~\ref{prop:wate} provides one justification for the estimand $\bar{\tau}$: under positivity, it is a weighted average of the CATE over the observational dataset.
\begin{proposition}
\label{prop:wate}
Assume that $1 \in \psi(\cdot)$ and Assumptions~\ref{assump:dominance} and~\ref{assump:positivity} hold with
\[
\Pr_{\rct}(c^{\top} \psi(\Delta(X))=0) < 1, \quad \forall c \in \real^p \setminus \{0\}.
\]
Then there exists $w(\cdot)$ with $\e_{\obs}[w(X)]=1$ such that 
\[
\bar{\tau} = \e_{\obs}[w(X)\mu(X)].
\]
\end{proposition}
\begin{proof}
Let $\Lambda(x)=f_{\rct}(x)/f_{\obs}(x)$ 
and note that
\begin{align*}
& \e_{\obs}\left[\mu(X)\left(1-\Lambda(X) \cdot [\mu(X)-\bar{\mu}(X)] \cdot \psi(\Delta(X))^{\top}\gamma\right)\right] \\
& \quad = \bar{\tau} + \e_{\obs}[\mu(X)-\bar{\mu}(X)] - \left(\e_{\obs}\left[\Lambda(X) \cdot \mu(X) \cdot (\mu(X)-\bar{\mu}(X)) \cdot \psi(\Delta(X))\right]\right)^{\top}\gamma.
\end{align*}
for all $\gamma \in \real^p$. It follows that if $\gamma$ satisfies
\begin{equation}
\label{eq:weight_constraint}
b^{\top}\gamma = \e_{\obs}[\mu(X)-\bar{\mu}(X)], \quad b = \e_{\obs}\left[\Lambda(X) \cdot \mu(X) \cdot (\mu(X)-\bar{\mu}(X))\psi(\Delta(X))\right]
\end{equation}
then $\bar{\tau}=\e_{\obs}[w(X)\mu(X)]$ for
\begin{equation}
\label{eq:w}
w(x) = 1-\Lambda(x) \cdot (\mu(x)-\bar{\mu}(x)) \cdot \psi(\Delta(x))^{\top}\gamma,
\end{equation}
As the first-order conditions defining $\bar{\mu}(\cdot)$ ensure that
\begin{equation}
\label{eq:mu_bar_foc}
\e_{\rct}[(\mu(X)-\bar{\mu}(X))\psi(\Delta(X))] = 0,
\end{equation}
we have
\[
0 = \left(\e_{\rct}\left[(\mu(X)-\bar{\mu}(X)) \cdot \psi(\Delta(X))\right]\right)^{\top}\gamma=\left(\e_{\obs}\left[\Lambda(X) \cdot (\mu(X)-\bar{\mu}(X)) \cdot \psi(\Delta(X))\right]\right)^{\top}\gamma
\]
and hence $\e_{\obs}[w(X)]=1$.

Thus, we have completed the proof so long as there exists $\gamma$ satisfying~\eqref{eq:weight_constraint}. 
So long as $b \neq 0$,
there will be a $(p-1)$-dimensional affine subspace of parameters $\gamma \in \real^p$ satisfying the equality constraint~\eqref{eq:weight_constraint}.
In the case that $b=0$,
since $1 \in \psi(\cdot)$ (i.e., $\psi(\cdot)$ contains an intercept term) we must have
\[
0 = \e_{\obs}[\Lambda(X) \cdot \mu(X) \cdot (\mu(X)-\bar{\mu}(X))] = \e_{\rct}[\mu(X) \cdot (\mu(X)-\bar{\mu}(X))] =\e_{\rct}[(\mu(X)-\bar{\mu}(X))^2]
\]
where the last equality follows by the first order conditions~\eqref{eq:mu_bar_foc}.
Hence $b=0$ implies $\mu(X)=\bar{\mu}(X)$ w.p.1. in the RCT,
which further implies $\mu(X)=\bar{\mu}(X)$ w.p.1. in the observational dataset since $\bar{\mu}$ is uniquely defined by~\eqref{eq:mu_bar_explicit} when no components of $\psi(\Delta(X))$ are linearly dependent w.p.1. in the RCT. 
Thus, in the case $b=0$, Proposition~\ref{prop:wate} holds simply with $w(x)=1$, i.e., in that case $\bar{\tau}=\tau$.
\end{proof}

As shown in the proof of Proposition~\ref{prop:wate} above,
the weights on the CATE defining $\bar{\tau}$ typically depend on the sampling propensity.
This is common in the literature on estimating the ATE in a single unconfounded observational study under positivity violations;
a leading example is given by the ``overlap" weights proportional to $e(x)(1-e(x))$,
which can be estimated with minimal variance among all propensity weighted-averages of CATEs~\citep{li2018balancing, li2019addressing,wang2025ratedoublyrobustestimation}.
It is easy to see from Proposition~\ref{prop:wate} that our estimand
$\bar{\tau}$ inherits a desirable property of these propensity-weighted average treatment effects in the causal transportation setting: it reduces to the unweighted average treatment effect $\tau$ when the (sampling) propensity is constant,
i.e., $f_{\rct}(\cdot)=f_{\obs}(\cdot)$,
so that there is no covariate shift between the observational and experimental datasets.
To see this, note that 
the first order conditions for $\bar{\mu}$ given by~\eqref{eq:mu_bar_foc}
imply $\e_{\rct}[\mu(X)]=\e_{\rct}[\bar{\mu}(X)]$ when $1 \in \psi(\cdot)$.
Thus, when additionally $f_{\rct}(\cdot)=f_{\obs}(\cdot)$,
we have
\[
\bar{\tau}=\e_{\obs}[\bar{\mu}(X)]=\e_{\rct}[\bar{\mu}(X)]=\e_{\rct}[\mu(X)]=\e_{\obs}[\mu(X)]=\tau.
\]

Like any WATE,
$\bar{\tau}$ is also evidently equal to $\tau$ when the CATE $\mu(\cdot)$ is constant.
By construction, there is an additional reasonable condition under which $\bar{\tau}=\tau$: when the CATE $\mu(\cdot)$ can be calibrated via a low-dimensional parametric correction of the observational conditional mean contrast $\Delta(\cdot)$, i.e., $\mu \in \cf$ so that $\bar{\mu}=\mu$.

\section{Estimation and inference}
\label{sec:estimation}

Our proposed estimation and inference procedure for the estimand $\bar{\tau}$ is simple and intuitive:
\begin{enumerate}
\item Partition the set of observational indices $\{n+1,\ldots,n+N\}$ randomly and as evenly as possible into $K$ folds $\ci_1,\ldots,\ci_K$ \label{item:cross-fitting}
\item Compute \emph{cross-fit} estimates $\hat{\Delta}^{(-1)}(\cdot),\ldots,\hat{\Delta}^{(-K)}(\cdot)$ of the treatment-control contrast $\Delta(\cdot)$ in the observational dataset, where $\hat{\Delta}^{(-k)}(\cdot)$ is computed by applying a CATE learner to all samples in the observational dataset \emph{outside} fold $k=1,\ldots,K$. Then let $\hat{\Delta}(\cdot) = K^{-1} \sum_{k=1}^K \hat{\Delta}^{(-k)}(\cdot)$. \label{item:Delta_estimate}
\item Compute estimates $\hat{\bar{\mu}}(\cdot)=\hat{\bar{\beta}}^{\top}\psi(\hat{\Delta}(\cdot))$ and $\hat{\bar{\mu}}^{(-k)}(\cdot) = \hat{\bar{\beta}}^{\top}\psi(\hat{\Delta}^{(-k)}(\cdot)), k=1,\ldots,K$ of $\bar{\mu}(\cdot)$ where $\hat{\bar{\beta}}$ estimates $\bar{\beta}$ in~\eqref{eq:mu_bar_explicit} by predicting $D_i$ from $\psi(\hat{\Delta}(X_i))$ in the experimental dataset using ordinary least squares \label{item:linear_calibration}
\item Average the out-of-fold estimates of $\bar{\mu}$ evaluated at the covariates in the observational dataset:
\begin{equation}
\label{eq:tau_bar_hat}
\hat{\bar{\tau}} = \frac{1}{N} \sum_{k=1}^K \sum_{i \in \ci_k} \hat{\bar{\mu}}^{(-k)}(X_i).
\end{equation}
\label{item:tau_bar_hat}
\item Return $\hat{\bar{\tau}}$ along with the confidence interval $\hat{\bar{\tau}} \pm z_{1-\alpha/2} \sqrt{\hat{V}_{n,N}}$ where
\begin{equation}
\label{eq:V_hat}
\hat{V}_{n,N} = \hat{a}_{n,N}^{\top} \left(\frac{1}{n^2} \sum_{i=1}^n \hat{r}_i\hat{r}_i^{\top}\right) \hat{a}_{n,N} + \frac{1}{N^2} \sum_{k=1}^K \sum_{i \in \ci_k} (\hat{\bar{\mu}}^{(-k)}(X_i)-\hat{\bar{\tau}})^2
\end{equation}
and $z_{1-\alpha/2}$ is the $1-\alpha/2$ quantile of the standard normal distribution with
\begin{align*}
\hat{a}_{n,N} & = \left(\frac{1}{n} \sum_{i=1}^n \psi(\hat{\Delta}(X_i))\psi(\hat{\Delta}(X_i))^{\top}\right)^{-1} \times \frac{1}{N} \sum_{k=1}^K\sum_{i \in \ci_k} \psi(\hat{\Delta}^{(-k)}(X_i)) \qquad \text {and} \\
\hat{r_i} & = (D_i-\hat{\bar{\mu}}(X_i))\psi(\hat{\Delta}(X_i)), \quad i=1,\ldots,n
\end{align*}
\label{item:ci}
\end{enumerate}

The above estimation steps introduce three sources of error into the final point estimate $\hat{\bar{\tau}}$: the \emph{observational contrast estimation error} from estimating $\Delta$ in step~\ref{item:Delta_estimate}, the \emph{calibration error} from estimating $\bar{\beta}$ in the OLS regression in step~\ref{item:linear_calibration},
and the \emph{observational sampling error} from the sample average approximating the population expectation in step~\ref{item:tau_bar_hat}.
The variance estimate~\eqref{eq:V_hat} for the final confidence interval only estimates the calibration error and ignores the other two sources of error.
Intuitively, this is asymptotically plausible as $n,N \rightarrow \infty$ with $N \gg n$, i.e., the size of experimental dataset is asymptotically negligible relative to the size of the observational dataset as both grow.
This is because the calibration error scales inversely with $n$ while the other two error sources scale inversely with $N$.

Doing a more careful analysis, under regularity assumptions,
the calibration error should be on the order of $n^{-1/2}$ by standard linear regression theory.
The observational sampling error is of order $N^{-1/2}$ which is negligible compared to the calibration error as long as $N \gg n$.
Finally, the observational contrast estimation error will decay more slowly than $N^{-1/2}$ if $\Delta$ is estimated nonparametrically.
The attainable nonparametric contrast error rate generally depends on the smoothness of $\Delta$ and/or the attainable rates for estimating the nuisance functions $m_z(x) = \e_{\obs}[Y \mid X=x,Z=z], z=0,1$.
As a concrete example,
if $X$ has compact support in the observational dataset,
the observational propensity score $r(x)=\e_{\obs}[Z \mid X=x]$ lies in $[\epsilon,1-\epsilon]$ with probability 1 for some $\epsilon > 0$,
and $m_0$ and $m_1$ are c\`adl\`ag with finite sectional variation norm,
then using the highly adaptive LASSO~\citep{benkeser2016highly} to estimate $m_0(\cdot)$ and $m_1(\cdot)$ separately and then taking their difference to estimate $\Delta(\cdot)$ should achieve a rate no slower than $N^{-1/3}$ under some additional mild regularity conditions.
Faster rates could be achieved with various smoothness assumptions on $\Delta(\cdot)$~\citep{kennedy2023towards}.
In any case,
for first-order asymptotic validity of our inference we require the nonparametric estimation error for $\Delta$ to decay more quickly than $n^{-1/2}$ so that the observational contrast estimation error is negligible compared to the calibration error.

Our proposal to use cross-fitting follows the literature on double/debiased machine learning~\citep{chernozhukov2018double} to avoid needing to impose Donsker conditions on $\Delta(\cdot)$ and the corresponding estimates $\hat{\Delta}(\cdot)$.
Letting $\|f\|_{q,\rct} = \left(\int_{\cx} |f|^q(x) f_{\rct}(x)d\lambda(x)\right)^{1/q}$ denote the $L^q$ norm of any function $f$ on $\cx$ with $\e_{\rct}[|f(X)|^q]<\infty$
(and analogously defining $\|f\|_{q,\obs}$),
Theorem~\ref{thm:estimation_and_inference} below formalizes all the necessary conditions for our estimator $\hat{\bar{\tau}}$ to be $n^{1/2}$-consistent for $\bar{\tau}$ and asymptotically normal,
enabling asymptotically valid inference with our proposed confidence intervals.

\begin{theorem}
\label{thm:estimation_and_inference}
Suppose Assumption~\ref{assump:dominance} holds, $\e[|D|^2]< \infty$, $\e_{\rct}[\psi(\Delta(X))\psi(\Delta(X))^{\top}]$ is invertible, and
$n,N \rightarrow \infty$ with $n/N = o(1)$.
Additionally, suppose that the cross-fit estimates $\hat{\Delta}^{(-1)}(\cdot),\ldots,\hat{\Delta}^{(-k)}(\cdot)$ of $\Delta(\cdot)$ satisfy the rate conditions 
\begin{equation}
\label{eq:psi_rate}
    |\psi_j \circ \hat{\Delta}-\psi_j \circ \Delta\|_{2,\rct} + \|\psi_j \circ \hat{\Delta}^{(-k)}- \psi_j \circ \Delta\|_{2,\obs} = o_p(n^{-1/2}), \qquad j=1,\ldots,p, \quad k=1,\ldots,K.
\end{equation}
Then
\[
\sqrt{n}(\hat{\bar{\tau}}-\bar{\tau}) \stackrel{d}{\rightarrow} \mathcal{N}(0,\Sigma)
\]
where for $\bar{\alpha} = \left(\e_{\rct}[\psi(\Delta(X))\psi(\Delta(X))^{\top}]\right)^{-1} \e_{\obs}[\psi(\Delta(X))]$ we define
\begin{equation}
\label{eq:Sigma}
\Sigma = \bar{\alpha}^{\top} \e_{\rct}\left[(D-\bar{\mu}(X))^2\psi(\Delta(X))\psi(\Delta(X))^{\top}\right] \bar{\alpha}.
 \nonumber
\end{equation}
If additionally $\|\psi_j \circ \hat{\Delta}-\psi_j \circ \Delta\|_{2+\epsilon,\rct} = o_p(1)$ for some $\epsilon>0$ and all $j=1,\ldots,p$ with $\e_{\rct}[|D|^4]<\infty$, $\e_{\rct}[\|\psi(\Delta(X))\|^k]<\infty$ for all $k < \infty$ and $\psi(\hat{\Delta}(X))$ having finite moments of all orders in the RCT with probability tending to 1, meaning that
\begin{equation}
\label{eq:all_orders}
\liminf \Pr\left(\int \left\|\psi(\hat{\Delta}(x))\right\|^k f_{\rct}(x) d\lambda(x) < \infty \text{ for all } k < \infty \right) = 1,
\end{equation}
then the confidence interval
covers asymptotically:
\[
\liminf \Pr\left(\hat{\bar{\tau}} - z_{1-\alpha/2}\sqrt{\hat{V}_{n,N}} \leq \tau \leq \hat{\bar{\tau}} + z_{1-\alpha/2}\sqrt{\hat{V}_{n,N}}\right) \geq 1-\alpha, \quad \forall \alpha \in (0, 1).
\]
\end{theorem}

In the crop rotation data, 
we have $n < N^{1/2}$.
In this case, the rate condition $\|\hat{\Delta}-\Delta\|_{2,\obs} = o_p(n^{-1/2})$ is satisfied as long as $\Delta$ can be estimated at a root mean square rate faster than $o_p(N^{-1/4})$.
This is the standard double machine learning rate,
and immediately implies~\eqref{eq:psi_rate} if $\psi(z)=(1,z)^{\top}$.

\section{Semiparametric efficiency}
\label{sec:discussion}
Much of the related work on causal transportation has worked under a ``nested" trial framework where the experimental and observational datasets are viewed as i.i.d. samples from a single combined population,
instead of as independent samples from two different populations.
This enables the application of semiparametric efficiency theory to evaluate the optimality of proposed estimators.
To study the efficiency properties of our proposed estimator $\hat{\bar{\tau}}$,
we now re-formulate our setting in the nested trial framework.
We refer the reader to~\citet{colnet2024causal} for a more comprehensive treatment of the distinction between the nested and non-nested perspectives.

Consider the single ``fused" dataset $\{W_i=(Q_i,X_i,Y_i)\}_{i=1}^M$ of size $M = n+N$
that concatenates the experimental and observational datasets.
Here, $D_i$ in the experimental dataset is renamed to $Y_i$ (cf.~\citet{li2024efficient})
and $Q_i=2$ for all observations $i=1,\ldots,n$ originally in the experimental dataset
while $Q_i \equiv Z_i \in \{0,1\}$ for all observations $i=n+1,\ldots,n+N$ originally in the observational dataset.
Then, for example, the conditional distribution of $Y$ given $Q=2$ under the nested framework can be viewed as equivalent to the marginal distribution of $D$ in the experiment under our non-nested framework from the previous sections.
More generally,
the distribution of the experimental (resp. observational) dataset
is reinterpreted as a conditional distribution given $Q=2$ (resp. $Q \neq 2$).
For instance,
\begin{equation}
\label{eq:Delta_nested}
\Delta(x) = m_1(x)-m_0(x); \qquad m_q(x) = \e[Y \mid Q=q], \quad q=0,1,2,
\end{equation}
while the asymptotic variance $\Sigma$ from~\eqref{eq:Sigma} can be reinterpreted as
\begin{equation}
\label{eq:Sigma_nested}
\Sigma = \bar{\alpha}^{\top} \e\left[(Y-\bar{\mu}(X))^2\psi(\Delta(X))\psi(\Delta(X))^{\top} \mid Q=2\right] \bar{\alpha}
\end{equation}
for
\begin{equation}
\label{eq:alpha_nested}
\bar{\alpha} = \left(\e[\psi(\Delta(X))\psi(\Delta(X))^{\top} \mid Q = 2]\right)^{-1}\e[\psi(\Delta(X)) \mid Q \neq 2],
\end{equation}
and we write the likelihood ratio from the proof of Proposition~\ref{prop:wate} as
\begin{equation}
\label{eq:Lambda}
\Lambda(x) = \frac{f_{X \mid Q=2}(x)}{f_{X \mid Q \neq 2}(x)}
\end{equation}
where $f_{X \mid Q=2}(x) \equiv f_{\rct}(x)$ and $f_{X \mid Q \neq 2}(x) \equiv f_{\obs}(x)$ are conditional densities of $X$ given $Q=2$ (resp. $Q \neq 2$). Additionally, we have the observational propensity score
\begin{equation}
\label{eq:r_nested}
r(x) = \Pr(Q=1 \mid X=x,Q \neq 2).
\end{equation}

We now derive and state the semiparametric efficiency bound for our projected estimand $\bar{\tau}$
under the nested formulation of our setup with the single fused dataset $\{(Q_i,X_i,Y_i)\}_{i=1}^M$ and $M \rightarrow \infty$.
The estimand is
\[
\bar{\tau} = \e[\bar{\mu}(X) \mid Q \neq 2], \quad \bar{\mu} = \argmin_{f \in \cf} \e[(Y-f(X))^2 \mid Q=2] = \argmin_{f \in \cf} \e[\indic(Q=2)(Y-f(X))^2].
\]
Assuming that 
\begin{equation}
\label{eq:nested_linear_independence}
\Pr(c^{\top}\psi(\Delta(X))=0 \mid Q=2) = 0, \quad \forall c \in \real^p \setminus \{0\}
\end{equation}
(the nested analogue of the linear independence condition in Proposition~\ref{prop:wate}), we can write
\begin{equation}
\label{eq:tau_bar_nested}
\bar{\mu}(x) = \psi(\Delta(x))^{\top} \bar{\beta}, \quad 
\bar{\tau} = \e[\psi(\Delta(X)) \mid Q \neq 2]^{\top}\bar{\beta}
\end{equation}
where $\bar{\beta} = \left(\e[\psi(\Delta(X))\psi(\Delta(X))^{\top} \mid Q=2]\right)^{-1}\e\left[Y\psi(\Delta(X)) \mid Q=2\right]$ is the nested analogue of~\eqref{eq:mu_bar_explicit}.
Then, we let
\begin{equation}
\label{eq:rho_q_nested}
\rho_q = \Pr(Q = q), \quad q=0,1,2,
\end{equation}
and follow~\citet{zhang2023double} and~\citet{kallus2025role}
in examining the behavior of our efficiency bound as $\rho_2  \rightarrow 0$ as $M \rightarrow \infty$ but the conditional distribution of $W=(Q,X,Y)$ given $\indic(Q=2)$ remains constant.
In words, this is the asymptotic regime where the distributions within each dataset remain constant while the ratio of the experimental sample size to the observational sample size vanishes.
This accounts for our nonstandard setting with the experimental dataset being small compared to the observational dataset.
Our main result is that in this regime,
the semiparametric efficiency bound --- assuming it is finite --- approaches the asymptotic variance $\Sigma$ of our proposed outcome regression estimator $\hat{\bar{\tau}}$.
Thus,
when the conditions of Theorem~\ref{thm:estimation_and_inference} are satisfied,
our estimator $\hat{\bar{\tau}}$ is efficient for $\bar{\tau}$.

\begin{theorem}
\label{thm:efficiency}
Suppose $\{W_i=(Q_i,X_i,Y_i)\}_{i=1}^M$ are i.i.d. with the support of $Q_i$ being $\{0,1,2\}$, equation~\eqref{eq:nested_linear_independence} holds, and $\psi(\cdot)$ is continuously differentiable on $\cx$ with gradient $\dot{\psi}(\cdot)$.
Then the efficient influence function (as $M \rightarrow \infty$) for the estimand $\bar{\tau}$ in~\eqref{eq:tau_bar_nested} is
\[
\psi(w;\bar{\tau},\eta) = \frac{\indic(q \neq 2)}{1-\rho_2}\left(\bar{\mu}(x)-\bar{\tau}+\iota(w;\eta)\kappa(x;\eta)\right) + \frac{\indic(q=2)}{\rho_2}(y-\bar{\mu}(x))\bar{\alpha}^{\top}\psi(\Delta(x))
\]
where $\eta$ consists of the nuisance functions and parameters defined in equations~\eqref{eq:Delta_nested} --~\eqref{eq:rho_q_nested} and $w=(q,x,y)$ with
\begin{align*}
\iota(w;\eta) & = \frac{\indic(q=1)}{r(x)}(y-m_1(x)) - \frac{\indic(q=0)}{1-r(x)}(y-m_0(x)), \quad \text{ and } \\
\kappa(x;\eta) &= \dot{\psi}(\Delta(x))^{\top}\bar{\beta} + \Lambda(x)\bar{\alpha}^{\top}\left(m_2(x)\dot{\psi}(\Delta(x))-\left[\psi(\Delta(x))\dot{\psi}(\Delta(x))^{\top}+\dot{\psi}(\Delta(x))\psi(\Delta(x))^{\top}\right]\right)\bar{\beta}.
\end{align*}
Hence, the semiparametric efficiency bound is $V_{\eff}=\e[\psi^2(W;\bar{\tau},\eta)]$.
If $V_{\eff}<\infty$,
then as $M \rightarrow \infty$ with $\rho_2 \rightarrow 0$ while the conditional distribution of $(Q,X,Y)$ given $\indic(Q = 2)$ remains constant,
we have $\rho_2 V_{\eff} \rightarrow \Sigma$ for $\Sigma$ as in~\eqref{eq:Sigma_nested}.
\end{theorem}

It is instructive to contrast the efficiency bound for $\bar{\tau}$,
derived in Theorem~\ref{thm:efficiency},
with the efficiency bound for the transported average treatment effect $\tau$,
given in~\citet{li2023note} and other references.
An important quantity in the efficiency bound for $\tau$ is the experimental sampling propensity score $\pi(x)=\Pr(Q=2 \mid X=x)$.
By Bayes' rule, we write
\begin{equation}
\label{eq:pi}
\pi(x) = \frac{f_{X \mid Q=2}(x)\Pr(Q=2)}{f_{X \mid Q = 2}(x)\Pr(Q=2)+f_{X \mid Q \neq 2}(x)\Pr(Q \neq 2)} = \frac{\rho_2 \Lambda(x)}{\rho_2 \Lambda(x)+(1-\rho_2)}.
\end{equation}
The efficiency bound for $\tau$ is infinite if $\e[\pi(X)^{-1} \mid Q \neq 2]=\infty$.
By~\eqref{eq:pi},
this means we can only hope to have a finite efficiency bound for $\tau$ if
$\e[\Lambda(X)^{-1} \mid Q \neq 2]<\infty$.
This requires a strong form of positivity (Assumption~\ref{assump:positivity}).
By contrast,
the only term depending on $\Lambda(\cdot)$ in the efficient influence function $\psi$ for $\bar{\tau}$ derived in Theorem~\ref{thm:efficiency}
is
\[
\frac{\indic(q \neq 2)}{(1-\rho_2)}\iota(w;\eta)\kappa(x;\eta)
\]
which has finite squared expectation if we assume $\psi(\Delta(X))$, $\dot{\psi}(\Delta(X))$, and $Y$ have finite moments of all orders in the observational dataset and $\Lambda(X)$ (not $\Lambda(X)^{-1}$) has a finite moment of order strictly greater than 2.

We note that prior works~\citep{newey1994asymptotic, chernozhukov2021automatic, hirshberg2021augmented} have studied semiparametric efficient estimation for projection estimands of the form
\begin{equation}
\label{eq:projection_estimand}
\bar{\theta} = \e[m(W,\bar{\gamma})], \quad \bar{\gamma}(\cdot) = \argmin_{\gamma \in \Gamma} \e[\ell(W,\gamma)]
\end{equation}
in a general supervised learning setting with i.i.d. observations $\{W_i=(Y_i,R_i)\}_{i=1}^M$,
where $\ell(\cdot,\cdot)$ is some loss function and $\Gamma$ is any function collection that is linear and closed with respect to the Hilbert space $\ch_R$ of all square integrable functions of $R$.
In our setting,
taking $R=(Q,X)$,
the estimand $\bar{\theta} = (1-\rho_2)\bar{\tau}$ takes the form~\eqref{eq:projection_estimand}
with $\bar{\gamma}=\bar{\mu}$, 
$\ell(W,\bar{\gamma}) = \indic(Q=2)(Y-\bar{\gamma}(X))^2$,
$\Gamma=\cf$,
and $m(W,\bar{\gamma})= \indic(Q \neq 2)\bar{\gamma}(X)$.
Then the cited works give the efficiency bound for $\bar{\theta}$,
which can be easily converted into an efficiency bound for $\bar{\tau}=\bar{\theta}/(1-\rho_2)$
(Appendix~\ref{app:eff_bound}).
However, these efficiency bounds are based on the assumption that the function class $\Gamma$ is known.
In our setting,
this is equivalent to $\Delta(\cdot)$ being known.
It turns out that the efficiency bound for $\bar{\tau}$ when $\Delta(\cdot)$ is known
is generally smaller than the bound $V_{\eff}$ in Theorem~\ref{thm:efficiency}.
A similar result was shown recently by~\citet{wang2025ratedoublyrobustestimation} for estimating general weighted average treatment effects with weights depending on the propensity score.
However, the difference between the efficiency bound for $\bar{\tau}$ when $\Delta(\cdot)$ is known and the bound when $\Delta(\cdot)$ must be estimated tends to zero as $\rho_2 \rightarrow 0$,
analogous to what~\citet{kallus2025role} found in the context of semi-supervised learning with surrogate outcomes.
This formalizes our intuition that when the observational dataset is much larger than the experimental dataset,
we can ignore the observational contrast estimation error.
Nonetheless, our novel efficiency bound $V_{\eff}$ and corresponding efficient influence function from Theorem~\ref{thm:efficiency} are useful for constructing an efficient one-step or targeted maximum likelihood estimator for $\bar{\tau}$
when $\rho_2 \not\rightarrow 0$ and $\Delta(\cdot)$ is not known.
We do not work through the details here,
focusing instead on our simpler outcome regression estimator $\hat{\bar{\tau}}$ which is asymptotically equivalent (under the conditions of Theorem~\ref{thm:estimation_and_inference}) in the regime of interest with $\rho_2 \rightarrow 0$.

\section{Numerical studies}
\label{sec:simulations}

Through numerical studies on two sets of simulated data-generating processes and one real dataset combining field experiments and satellite-based crop yield estimates to estimate the impact of crop rotation on maize yields,
we illustrate the numerical performance of our proposed estimator $\hat{\bar{\tau}}$ (and our corresponding confidence intervals) relative to two alternatives in the causal transportation literature.
The alternative estimators are:
\begin{enumerate}
\item The augmented inverse propensity of sampling weighting (AIPSW) estimator, given by
\begin{equation}
\label{eq:aipsw}
\hat{\tau}_{\aipsw} = \frac{1}{n} \sum_{i=1}^n \frac{n}{N} (D_i-\hat{\mu}(X_i)) \hat{q}_{n,N}(X_i) + \frac{1}{N} \sum_{i=n+1}^{n+N} \hat{\mu}(X_i)
\end{equation}
where $\hat{q}_{n,N}(x)$ estimates the odds $q_{n,N}(x)$ that a randomly selected observation with $X=x$ in the combined observational and experimental dataset is in the observational dataset:
\[
q_{n,N}(x) = \frac{N}{n} \frac{f_{\obs}(x)}{f_{\rct}(x)}.
\]
See~\citet{colnet2024causal} for an extensive discussion of the AIPSW estimator.
\item The collaborative estimator of~\citet{rudolph_improving_2025}, given by
\begin{align}
\hat{\tau}_{\collab} & = \frac{1}{n}\sum_{i=1}^n \frac{n}{N}\left[(D_i-\hat{\mu}(X_i))\hat{g}_{n,N}(X_i) + \frac{\hat{q}_{n,N}(X_i)}{1+\hat{q}_{n,N}(X_i)}(\hat{\mu}(X_i)-\hat{k}_{n,N}(X_i))\right] \nonumber \\
& \quad + \frac{1}{N} \sum_{i=n+1}^{n+N} \frac{\hat{q}_{n,N}(X_i)\hat{\mu}(X_i) + \hat{k}_{n,N}(X_i)}{1+\hat{q}_{n,N}(X_i)} \label{eq:collab}
\end{align}
where $\hat{g}_{n,N}$ and $\hat{k}_{n,N}$ estimate the functions
\[
g_{n,N}(x) = \frac{N}{n} \frac{f_{\obs}^{\mu}(\mu(x))}{f_{\rct}^{\mu}(\mu(x))} \quad \text{and} \quad k_{n,N}(x) = \e[\mu(X) \mid g_{n,N}(X)=g_{n,N}(x)],
\]
respectively,
where $f^{\mu}_{\rct}$ and $f^{\mu}_{\obs}$ are the densities of $\mu(X)$ in the experimental and observational datasets, respectively.
For intuition, note that $g_{n,N}(x)$ corresponds to the odds that a randomly selected observation with CATE $\mu(X)=\mu(x)$ in the combined observational and experimental dataset is in the observational dataset.
\end{enumerate}

The AIPSW estimator is well known to be regular and semiparametric efficient for estimating the transported average treatment effect $\tau$ under positivity.
The collaborative estimator is designed to improve upon the precision of AIPSW when the covariates are multivariate by replacing the scaled likelihood ratio $q_{n,N}$ with the univariate sampling odds propensity $g_{n,N}$,
and projecting out some of the residual variance in $\mu$ onto $k$.
However, 
the collaborative estimator is no longer regular
because $g_{n,N}$ and $k_{n,N}$ are conditional expectations given nuisance functions that must be estimated.
Asymptotic 95\% confidence intervals for $\tau$ based on the estimators $\hat{\tau}_{\aipsw}$ and $\hat{\tau}_{\collab}$ (valid under sufficiently strong positivity conditions) are given by
\[
\hat{\tau}_{\aipsw} \pm z_{0.975} \sqrt{\hat{V}_{\aipsw}}, \qquad \hat{\tau}_{\collab} \pm z_{0.975} \sqrt{\hat{V}_{\collab}}
\]
where the variance estimates $\hat{V}_{\aipsw}$ and $\hat{V}_{\collab}$ are computed using the estimated influence functions:
\begin{align*}
\hat{V}_{\aipsw} & = \frac{1}{N^2} \left(\sum_{i=1}^{n} \hat{q}_{n,N}^2(X_i)(D_i-\hat{\mu}(X_i))^2 + \sum_{i=n+1}^{n+N} (\hat{\mu}(X_i)-\hat{\tau}_{\aipsw})^2\right) \\
\hat{V}_{\collab} & = \frac{1}{N^2}\sum_{i=1}^n \left[\hat{g}_{n,N}(X_i)(D_i-\hat{\mu}(X_i))+ \frac{\hat{q}_{n,N}(X_i)}{1+\hat{q}_{n,N}(X_i)}(\hat{\mu}(X_i)-\hat{k}_{n,N}(X_i))\right]^2 \\
& + \frac{1}{N^2} \sum_{i=n+1}^{n+N} \left(\frac{\hat{q}_{n,N}(X_i)\hat{\mu}(X_i) + \hat{k}_{n,N}(X_i)}{1+\hat{q}_{n,N}(X_i)}-\hat{\tau}_{\collab}\right)^2
\end{align*}

Throughout our numerical studies,
we implement our proposed outcome regression estimator $\hat{\bar{\tau}}$ using the simple basis expansion $\psi(\Delta(x))=(1,\Delta(x))^{\top}$.
Appendix~\ref{app:nuisance_estimation} provides details on how all nuisance functions are estimated for all of the numerical studies.
For simplicity, throughout this section we use the notation of our original non-nested setup introduced in Section~\ref{sec:estimand},
rather than the notation of the nested trial framework introduced in Section~\ref{sec:discussion}.
Code and data to fully reproduce all results are publicly available at \url{https://github.com/hli90722/calibrating_large_scale_observations}.

\subsection{Univariate covariate simulation}
\label{sec:1d_sim}
We begin by examining a data generating process based on the simulation study of~\citet{kallus2018removing} with a single univariate covariate $X$.
That study considers a clear positivity violation: in the experimental dataset,
$X$ is uniformly distributed on the interval $(-1,1)$,
whereas in the observational dataset, $X$ follows a standard normal distribution.
In such a setting, existing causal transportation estimators like $\hat{\tau}_{\aipsw}$ and $\hat{\tau}_{\collab}$ that do not allow for any extrapolation would fail due to the inability to identify $\mu(x)$ for $x$ outside $(-1,1)$.
Thus, we instead consider a data generating process where positivity holds, but there is potentially significant covariate shift between the experimental and observational datasets.
Specifically, we let the covariate $X$ in the experimental dataset follow a normal distribution with mean $-\theta$ and variance $1-\theta$,
and vary $\theta \in \{0,0.1,0.2,0.3,0.4,0.5,0.6,0.7\}$;
evidently higher values of $\theta$ correspond to greater covariate shift.
All other aspects of the data-generating process --- including the standard normal distribution for $X$ in the observational dataset --- are unchanged from~\citet{kallus2018removing}.
Translating their setting into our notation,
we have observational contrast $\Delta(x) = 0.75x^2+3x+1$ and true CATE $\mu(x)=\Delta(x)-x$,
and the observational treatment assignments $Z_i$ are generated by fair coin flips independently of the covariates $X_i$.
We fix the experimental and observational sample sizes at $n=100$ and $N=10,000$, respectively.

\begin{figure}[!htb]
\begin{center}
\includegraphics[width=\linewidth]{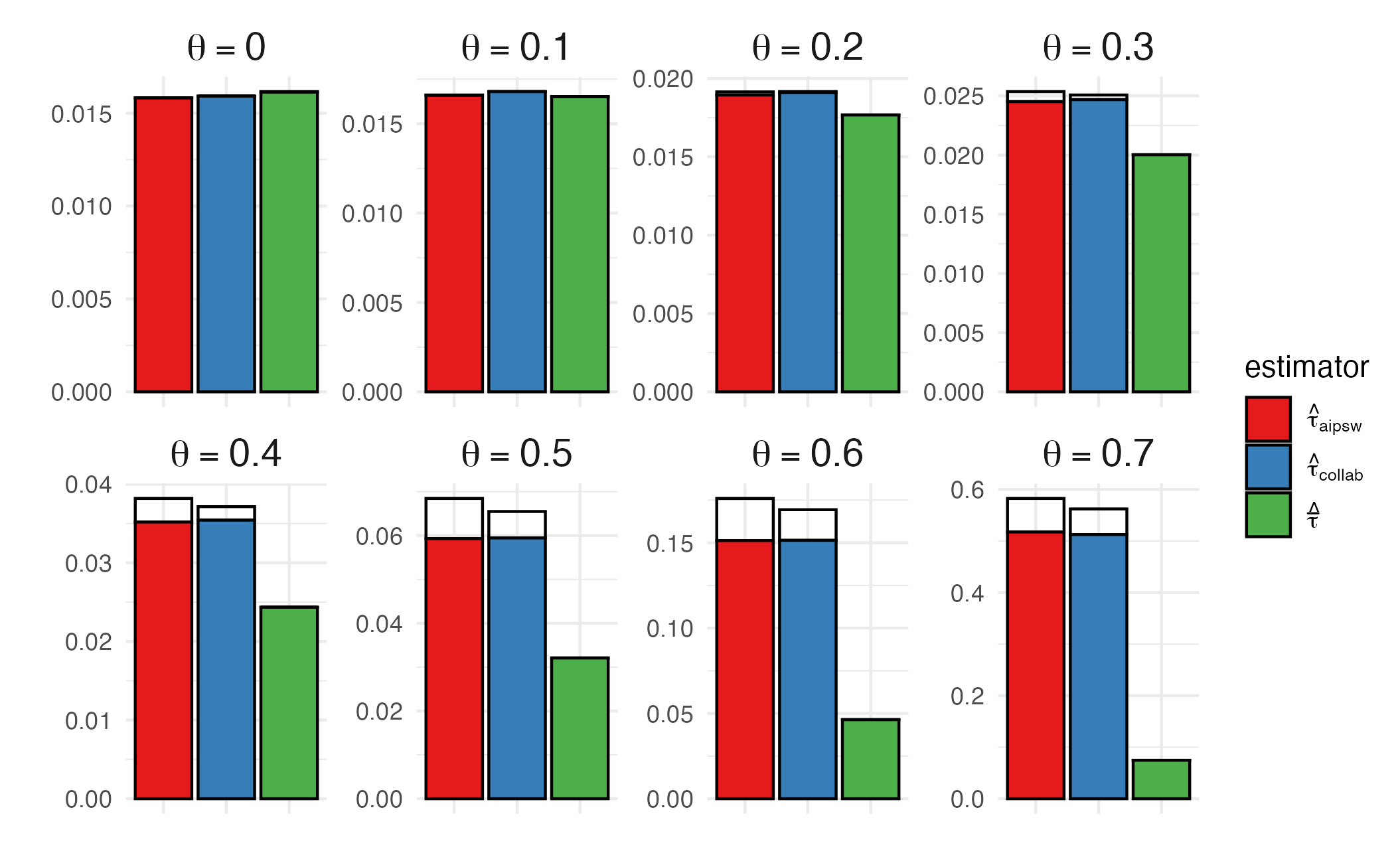}
\end{center}
\caption{The height of each bar corresponds to the estimated mean squared error (MSE) of an estimator from the simulation study of Section~\ref{sec:1d_sim},
based on 1000 simulations.
Each panel corresponds to a different value of the parameter $\theta$, which quantifies the degree of covariate shift between the experimental and observational datasets.
The estimators $\hat{\tau}_{\aipsw}$ and $\hat{\tau}_{\collab}$ are compared to the true transported average treatment effect $\tau$ while the estimator $\hat{\bar{\tau}}$ is compared to the projected treatment effect $\bar{\tau}$.
The shaded part of each bar corresponds to the estimated variance while the unshaded part of each bar corresponds to the estimated squared bias.}
\label{fig:mse_plot_1d_sim}
\end{figure}

Fig.~\ref{fig:mse_plot_1d_sim} shows that the collaborative and AIPSW estimators exhibit similar empirical mean squared errors (MSE) across all values of $\theta$.
This is unsurprising since the covariate is univariate and thus the collaborative estimator does not provide any dimensionality reduction advantages.
For the smallest values of $\theta$,
where the observational and experimental covariate distributions are similar, 
our outcome regression estimator $\hat{\bar{\tau}}$ performs similarly to the AIPSW and collaborative estimators.
As covariate shift increases,
however,
the outcome regression estimator $\hat{\bar{\tau}}$ shows increasing improvement over the AIPSW and collaborative estimators.
Indeed, for $\theta=0.7$, the MSE of $\hat{\bar{\tau}}$ (for estimating our estimand $\bar{\tau}$) is about 87\% lower than the MSE of the AIPSW and collaborative estimators (for estimating the transported ATE $\tau$).
Additionally, for the larger values of $\theta$,
we observe non-negligible bias (for the transported ATE $\tau$) in the AIPSW and collaborative estimators while the bias in our outcome regression estimator (for our projection estimand $\bar{\tau}$) remains sufficiently negligible as to be invisible in Fig.~\ref{fig:mse_plot_1d_sim}.

The quality of inference for $\tau$ using the AIPSW and collaborative estimators also degrades sharply as $\theta$ increases (Table~\ref{table:1d_ci_coverage}).
By contrast,
the empirical coverage of the asymptotic 95\% confidence intervals based on $\hat{\bar{\tau}}$ (for the estimand $\bar{\tau}$) stays in the range of about 92\%-93\% for all values of $\theta$,
whereas the coverage for the other two intervals (for the estimand $\tau$) degrades steadily.
Indeed, when $\theta \geq 0.5$,
\emph{oracle} variants of the AIPSW and collaborative estimators (that know the true nuisance functions) have infinite variance,
so the corresponding confidence intervals for $\tau$ are no longer asymptotically justified,
though it is clear from our simulations that the intervals' coverage breaks down for smaller values of $\theta$,
even in this rather benign simulation setting with a single covariate.

\begin{table}[!ht]
\centering
\caption{
The empirical coverage probabilities (left) and mean widths (right) for the nominal 95\% confidence intervals considered in the simulation study of Section 5.1,
across 1000 simulations.
}
\label{table:1d_ci_coverage}
\begin{minipage}{.45\linewidth}
\centering
\begin{tabular}{|c|c|c|c|}
\hline
$\theta$ & $\hat{\bar{\tau}}$ & $\hat{\tau}_{\aipsw}$ & $\hat{\tau}_{\collab}$ \\
\hline
0 & 92.8\% & 90.3\% & 90.2\% \\
0.1 & 92.6\% & 89.4\% & 88.9\% \\
0.2 & 92.7\% & 87.7\% & 88.3\% \\
0.3 & 92.2\% & 84.1\% & 85.1\% \\
0.4 & 92.1\% & 75.8\% & 78.1\% \\
0.5 & 91.6\% & 63.2\% & 66.9\% \\
0.6 & 91.9\% & 48.3\% & 51.6\% \\
0.7 & 92.1\% & 32.7\% & 36.3\% \\
\hline
\end{tabular}
\end{minipage}%
\hfill 
\begin{minipage}{.45\linewidth}
\centering
\begin{tabular}{|c|c|c|c|}
\hline
$\theta$ & $\hat{\bar{\tau}}$ & $\hat{\tau}_{\aipsw}$ & $\hat{\tau}_{\collab}$ \\
\hline
0 & 0.461 & 0.425 & 0.425 \\
0.1 & 0.462 & 0.424 & 0.425 \\
0.2 & 0.474 & 0.427 & 0.431 \\
0.3 & 0.502 & 0.433 & 0.444 \\
0.4 & 0.552 & 0.445 & 0.466 \\
0.5 & 0.636 & 0.462 & 0.492 \\
0.6 & 0.772 & 0.474 & 0.515 \\
0.7 & 0.994 & 0.471 & 0.530 \\
\hline
\end{tabular}
\end{minipage}
\end{table}

Equation~\eqref{eq:w} specifies a collection of weight functions $w(\cdot)$, indexed by a parameter $\gamma$
such that $\bar{\tau}=\e_{\obs}[w(X)\mu(X)]$ with $\e_{\obs}[w(X)]=1$.
To gain intuition about these weights,
in Fig.~\ref{fig:wt_plot_1d_sim} we plot one such weight function for each value of $\theta$.
Specifically,
we take
\begin{equation}
\label{eq:gamma}
\gamma = \frac{\e_{\obs}[\mu(X)-\bar{\mu}(X)]}{b^{\top}A^{-1}b} \cdot A^{-1}b
\end{equation}
where
\[
 A = \e_{\rct}[\Lambda(X)(\mu(X)-\bar{\mu}(X))^2\psi(\Delta(X))\psi(\Delta(X))^{\top}]
\]
and $b$ is as in~\eqref{eq:weight_constraint}.
This choice minimizes $\var_{\obs}(w(X))$ among all possible choices of $\gamma$ satisfying~\eqref{eq:weight_constraint}.
We observe that the weights are bounded within $(0,2)$ for $x \in [-4,4]$ across all values of $\theta$,
and indeed bounded within $(0.75,1.25)$ when $\theta \leq 0.5$,
supporting the intepretability of the estimand $\bar{\tau}$ in these simulation scenarios.

\begin{figure}[!htb]
\begin{center}
\includegraphics[width=0.75\linewidth]{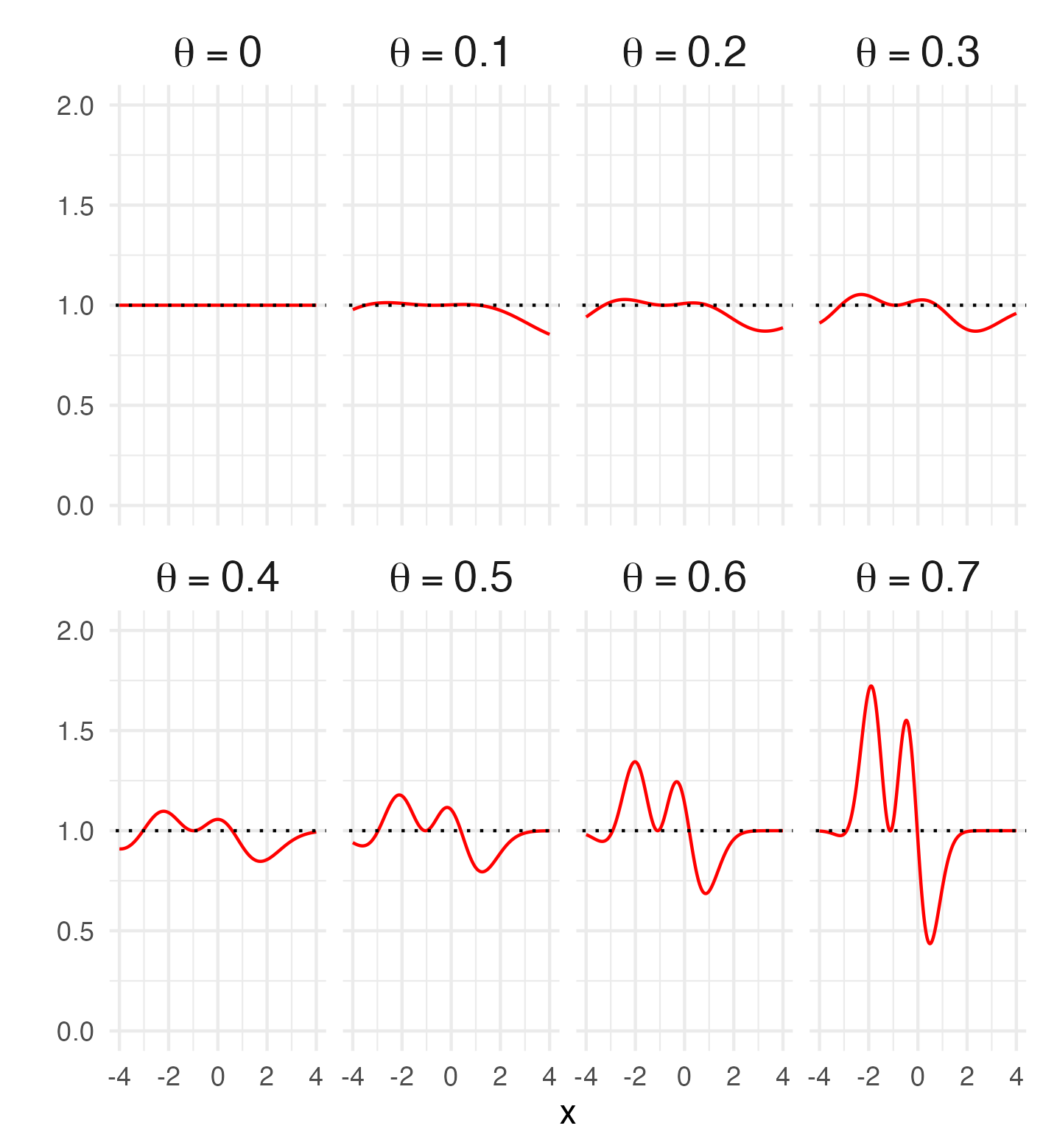}
\end{center}
\caption{For each value of $\theta$ studied in the simulations from Section~\ref{sec:1d_sim},
a plot of the weight function $w(\cdot)$ in~\eqref{eq:w} with $\gamma$ given by~\eqref{eq:gamma}.}
\label{fig:wt_plot_1d_sim}
\end{figure}

\subsection{Multivariate covariate simulation}
\label{sec:multivar_sim}

\begin{figure}[!htb]
\begin{center}
\includegraphics[width=\linewidth]{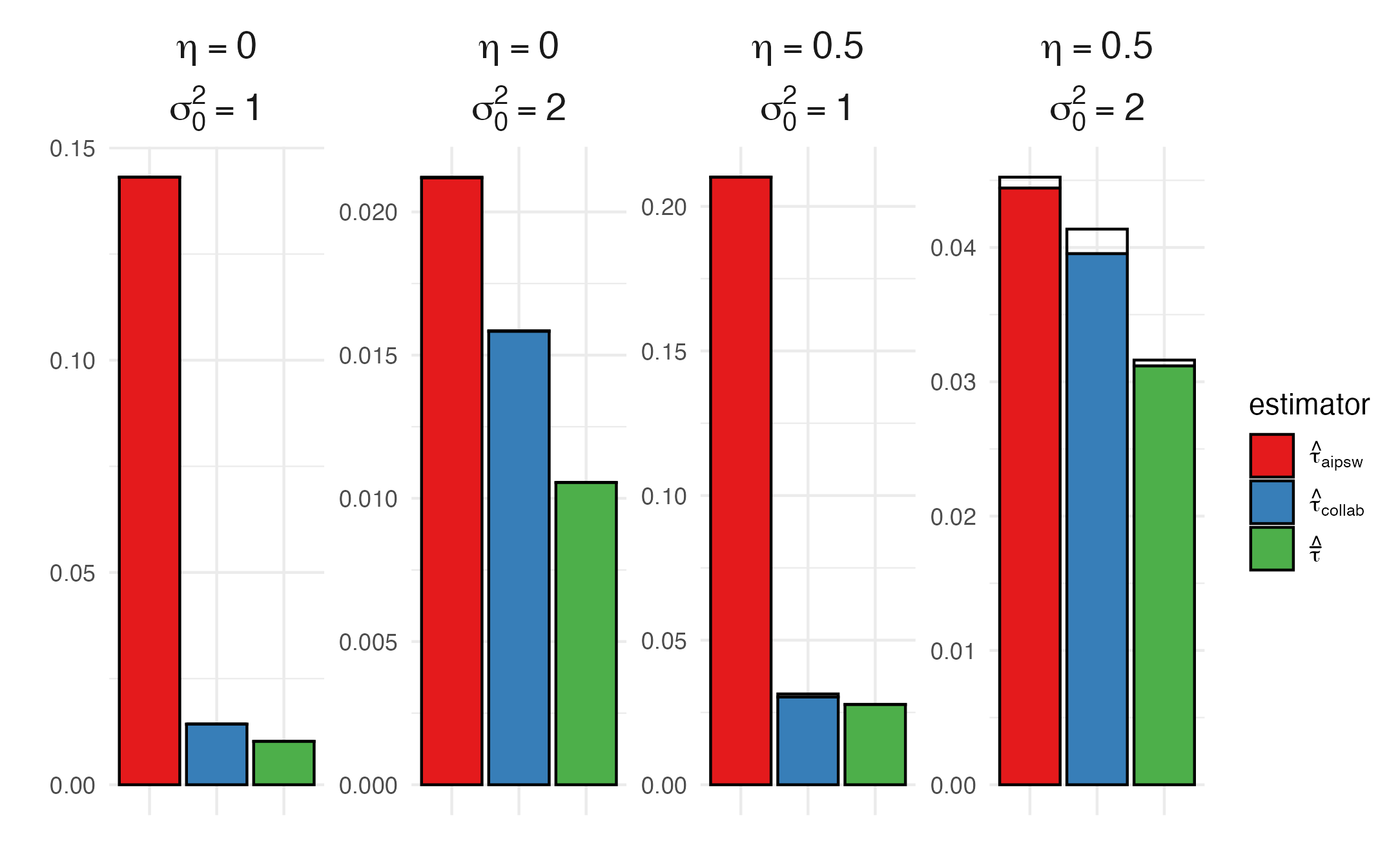}
\end{center}
\caption{Same as Fig.~\ref{fig:mse_plot_1d_sim}, but for the multivariate covariate simulations in Section~\ref{sec:multivar_sim}}
\label{fig:mse_plot_sim}
\end{figure}

\begin{table}[!htb]
\centering
\caption{
Same as Table~\ref{table:1d_ci_coverage} but for the multivariate covariate simulations in Section~\ref{sec:multivar_sim}
}
\label{table:ci_coverage}
\begin{minipage}{.45\linewidth}
\centering
\begin{tabular}{|c|c|c|c|c|}
\hline
$\eta$ & $\sigma_0^2$ & $\hat{\bar{\tau}}$ & $\hat{\tau}_{\aipsw}$ & $\hat{\tau}_{\collab}$ \\
\hline
0 & 1 & 94.9\% & 91.9\% & 91.9\% \\
0 & 2 & 94.9\% & 81.8\% & 89.9\% \\
0.5 & 1 & 92.5\% & 88.3\% & 80.3\% \\
0.5 & 2 & 90.8\% & 69.8\% & 76.4\% \\
\hline
\end{tabular}
\end{minipage}%
\hfill 
\begin{minipage}{.45\linewidth}
\centering
\begin{tabular}{|c|c|c|c|c|}
\hline
$\eta$ & $\sigma_0^2$ & $\hat{\bar{\tau}}$ & $\hat{\tau}_{\aipsw}$ & $\hat{\tau}_{\collab}$ \\
\hline
0 & 1 & 0.402 & 1.06 & 0.417 \\
0 & 2 & 0.406 & 0.396 & 0.419 \\
0.5 & 1 & 0.611 & 1.257 & 0.476 \\
0.5 & 2 & 0.616 & 0.453 & 0.486 \\
\hline
\end{tabular}
\end{minipage}
\end{table}

We now simulate four data generating processes (scenarios) with $d=10$ covariates.
For all scenarios, 
we have
\begin{align*}
\Delta(x) & = \expit(0.5x_1)\expit(1-0.5x_2) \\
\e_{\obs}[Y \mid X=x,Z=0] & = \sin(x_1) + \cos(x_2) + \sin(2x_1)\cos(2x_2)+0.5\sum_{i=1}^d x_i
\end{align*}
and observational propensity score
\[
\Pr_{\obs}(Z=1 \mid X=x) = \Phi\left(\frac{2x_2-x_1}{5}\right),
\]
where $\expit(x)=\exp(x)(1+\exp(x))^{-1}$ and $\Phi(\cdot)$ is the cumulative distribution function corresponding to the standard normal distribution.
We set 
\[
\mu(x) = 0.5 + 0.5\Delta(x) + \eta (x_1+1)(x_2+1)
\]
and have i.i.d. normally distributed covariates in the observational dataset with mean 0 and variance $\sigma_0^2$.
The four scenarios correspond to all combinations of $(\eta,\sigma_0^2)$ in $\{0, 0.5\} \times \{1,2\}$.
The covariates in the experimental dataset are i.i.d. normally distributed random variables with mean 0.5 and variance 1 for all scenarios.
As in Section~\ref{sec:1d_sim},
we fix the experimental and observational sample sizes at $n=100$ and $N=10,000$, respectively.

In all scenarios,
our point estimator $\hat{\bar{\tau}}$ exhibits the lowest MSE while the AIPSW estimator $\hat{\tau}_{\aipsw}$ exhibits the largest MSE (Fig.~\ref{fig:mse_plot_sim}).
Perhaps counterintuitively, the AIPSW estimator shows a substantially higher MSE in the scenarios where $\sigma_0^2=1$ than in those where $\sigma_0^2=2$,
even though the positivity violation is more severe in the latter case.
However, the collaborative estimator does not show the same pattern,
suggesting the explanation for this phenomenon lies in instability of nuisance estimation resulting from the curse of dimensionality.
In the two simulation scenarios with $\eta=0$, 
our linear regression calibration procedure is well-specified for the CATE and so $\tau=\bar{\tau} \approx 0.678$.
When $\eta=0.5$, the estimands differ significantly (with $\tau \approx 1.18$ and $\bar{\tau} \approx 1.79$).
Table~\ref{table:ci_coverage} shows that inference for $\tau$ (based on either the AIPSW or collaborative estimators) is considerably less reliable than our inference for $\bar{\tau}$,
especially when $\sigma_0^2=2$ (in which case the AIPSW and collaborative estimators have infinite asymptotic variance).
Our CI's based on $\hat{\bar{\tau}}$ also exhibit some degree of empirical undercoverage in the misspecified scenarios ($\eta=0.5$),
though in the well specified scenarios ($\eta=0$) they essentially achieve the nominal 95\% coverage unlike the other CI's.

\subsection{Crop rotation example}
\label{sec:data}

We now turn our attention to the dataset compiled by~\citet{kluger2022combining} to estimate the effect of soybean rotation on maize yields over a large agriculturally productive region of the Midwestern United States,
given by the extent of the colored points in Fig.~\ref{fig:map}.
The dataset contains 611 completely randomized blocked experiments conducted across $89$ (location, year) pairs under various growing conditions during the years 2000--2016,
sourced from~\citet{abendroth2017sustainable} and~\citet{bowles2020long}.
These are supplemented by about 1.31 million maize satellite-based yield estimates from the Scalable Crop Yield Mapper across different years at $174,869$ distinct locations (plotted in Fig.~\ref{fig:map}).
For each experiment,
the difference in average yields between the fields that grew soybeans in the prior year (i.e., were rotated) and the fields that grew maize in the prior year (i.e., were not rotated) was recorded.
We hereafter refer to this difference as the ``yield differential"; recall it corresponds to the variable $D_i$ in our notation.
A collection of covariates including year, latitude, longitude, and various weather and soil variables is available for all experiments and all satellite-based observations.
See Table 1 of~\citet{kluger2022combining} for more details.
Additionally, Section 2.1 of that paper provides additional information about the dataset.

We run 100 simulations.
Each simulation involves constructing a synthetic experimental dataset and a synthetic observational dataset by randomly sampling 10\% of the observations (without replacement) from the full experimental and observational datasets, respectively.
We compute each of the three estimators $\hat{\tau}_{\aipsw}$, $\hat{\tau}_{\collab}$, and $\hat{\bar{\tau}}$ (and the corresponding confidence intervals) using these synthetic datasets.
The true values of the estimands $\tau$ and $\bar{\tau}$ are taken to be the values of the estimators $\hat{\tau}_{\aipsw}$ and $\hat{\bar{\tau}}$,
respectively, on the full dataset.
They happen to be quite similar: we have $\tau=0.905$ and $\bar{\tau}=0.906$, respectively.
It is clear from Fig.~\ref{fig:crop_rotation_histograms} that the AIPSW estimator is extremely unstable,
with values spanning several orders of magnitude across the simulations.
This is not surprising given the clear positivity violation from the uneven spatial coverage of the experimental sites (Fig.~\ref{fig:map}).
The collaborative estimator fares quite a bit better, but the outcome regression $\hat{\bar{\tau}}$ is clearly much less variable and much more stable still (see Table~\ref{table:crop_rotation_coverage}).
This is perhaps unsurprising given the clear positivity violations.

\begin{figure}[!htb]
\begin{center}
\includegraphics[width=\linewidth]{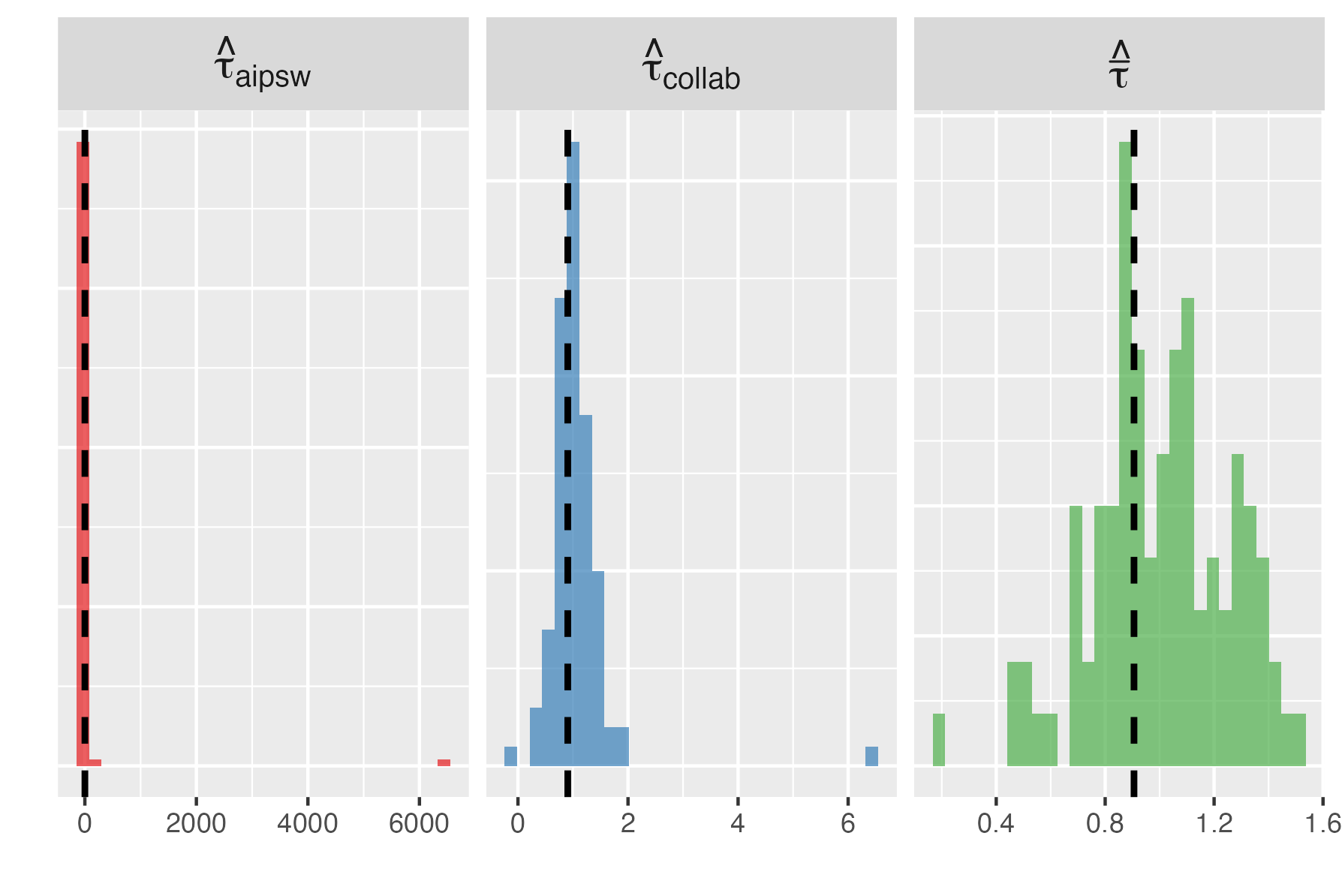}
\end{center}
\caption{Histograms of the three estimators across the 100 simulations from the crop rotation dataset described in Section~\ref{sec:data}. The black vertical dashed lines show the values of the true estimand --- the transported ATE $\tau \approx 0.905$ for the estimators $\hat{\tau}_{\aipsw}$ and $\hat{\tau}_{\collab}$,
and the projected ATE $\bar{\tau} \approx 0.906$ for the estimator $\hat{\bar{\tau}}$.}
\label{fig:crop_rotation_histograms}
\end{figure}

\begin{table}[!htb]
\centering
\caption{
The squared bias and variance of the three estimators across the 100 simulations from the crop rotation dataset described in Section~\ref{sec:data},
along with the empirical coverage rates and mean widths of the corresponding asymptotic 95\% confidence intervals. 
}
\label{table:crop_rotation_coverage}

\begin{tabular}{|c|c|c|c|}
\hline
& $\hat{\tau}_{\aipsw}$ & $\hat{\tau}_{\collab}$ & $\hat{\bar{\tau}}$  \\
\hline 
Squared bias & $4.43 \times 10^3$ & $0.0259$ & $0.00917$ \\
Variance & $4.15 \times 10^5$ & $0.397$ & $0.0630$ \\
Coverage & 98\% & 93\% & 96\% \\
Mean CI width & $269$ & $1.49$ & $0.985$ \\
\hline
\end{tabular}
\end{table}

\section{Conclusion}
\label{sec:conclusion}
We have proposed a method for transporting treatment effects in a setting with a small experimental dataset and a large observational dataset that has biased and/or noisy outcomes.
The key idea is to model treatment effect heterogeneity as a linear combination of a low dimensional basis expansion of a treatment-control contrast that can be readily learned from the large observational dataset.
When this modeling assumption holds,
we are able to consistently estimate the transported average treatment effect with a straightforward estimation procedure.
If the modeling assumption does not hold,
our procedure consistently estimates a transported estimand $\bar{\tau}$ that can be interpreted as a weighted average treatment effect under positivity,
analogous to the literature on estimating average treatment effects from a single unconfounded observational dataset under positivity violations.
Crucially, the estimand $\bar{\tau}$ remains identified without positivity.

Under regularity conditions, our proposed estimator achieves the efficiency bound for the transported estimand $\bar{\tau}$ when the observational dataset's size grows sufficiently quickly relative to the experimental dataset's size.
Most crucially for the practitioner,
the estimator does not require estimating any inverse propensity scores or weights, either directly or indirectly.
Instead, it only requires estimating the treatment-control contrast by applying any flexible machine learning method for the CATE to the large observational dataset, 
followed by applying ordinary least squares to the small experimental dataset.

Our methodology assumes independent and identically distributed observations.
This assumption forms the basis for our use of ordinary least squares for calibrating the observational dataset to the experimental dataset,
and the resulting characterization of the limiting estimand $\bar{\tau}$ as well as the proposed asymptotic inference methodology.
This is not necessarily realistic in settings like the crop rotation dataset of Section~\ref{sec:data},
where there are likely positive correlations in the outcomes between data points from the same location and/or year
that cannot be explained by observed covariates.
To account for this,~\citet{kluger2022combining} included random effects for year and site in their regressions calibrating the observational dataset to the experimental one.
It would be informative to carefully study inference and the limiting estimand for our estimator
(or propose alternative estimators)
under such alternative calibration procedures.

\section{Acknowledgments}
The author thanks Dan Kluger and David Lobell for sharing the experimental and observational dataset that inspired this work and forms the basis for the numerical study of Section~\ref{sec:data}. He is also grateful to Art Owen for early feedback on some of the ideas in this manuscript.

\bibliography{combine_obs_rct}
\bibliographystyle{abbrvnat}

\appendix

\section{Proofs}
\label{sec:proofs}

\subsection{Proof of Theorem~\ref{thm:estimation_and_inference}}

Throughout the proof we write
\begin{align*}
\psi_i & \equiv \psi(\Delta(X_i)), \quad i=1,\ldots,n+N \\
\hat{\psi}_i & \equiv \psi(\hat{\Delta}(X_i)), \quad i=1,\ldots,n \\
\hat{\psi}^{(-k)}_i & \equiv \psi(\hat{\Delta}^{(-k)}(X_i)), \quad i \in \ci_k, \quad k=1,\ldots,K
\end{align*}
for brevity.
We note that
\[
\sqrt{n}(\hat{\bar{\tau}}-\bar{\tau}) = \frac{\sqrt{n}}{N} \sum_{k=1}^K \sum_{i \in \ci_k}  \left(\hat{\bar{\beta}}^{\top}\hat{\psi}_i^{(-k)} - \bar{\beta}^{\top}\psi_i\right) + \frac{\sqrt{n}}{N} \sum_{i=n+1}^{n+N} \left(\bar{\mu}(X_i)-\bar{\tau}\right)
\]
where $\bar{\beta}$ is given by~\eqref{eq:mu_bar_explicit} since $\e_{\rct}[\psi(\Delta(X))\psi(\Delta(X))^{\top}]$ is invertible.
The rate conditions~\eqref{eq:psi_rate} imply $\e_{\obs}[\|\psi(\Delta(X))\|_2^2]<\infty$ and hence $\e_{\obs}[\bar{\mu}(X)^2]<\infty$,
so
\[
\frac{\sqrt{n}}{N}\sum_{i=n+1}^{n+N}(\bar{\mu}(X_i)-\bar{\tau}) = \sqrt{n}  \cdot \frac{1}{N}\sum_{i=n+1}^{n+N}(\bar{\mu}(X_i)-\bar{\tau}) =\sqrt{n} \cdot O_p(N^{-1/2}) = o_p(1)
\]
and it suffices to show that 
\[
\frac{\sqrt{n}}{N} \sum_{k=1}^K \sum_{i \in \ci_k}  \left(\hat{\bar{\beta}}^{\top}\hat{\psi}^{(-k)}_i - \bar{\beta}^{\top}\psi_i\right) \stackrel{d}{\rightarrow} \mathcal{N}(0,\Sigma).
\]
To that end,
we write
\begin{align*}
\frac{\sqrt{n}}{N} \sum_{k=1}^K \sum_{i \in \ci_k}  \left(\hat{\bar{\beta}}^{\top}\hat{\psi}_i^{(-k)} - \bar{\beta}^{\top}\psi_i\right) = \sqrt{n}(\hat{\bar{\beta}}-\bar{\beta})^{\top} \cdot \frac{1}{N} \sum_{k=1}^K \sum_{i \in \ci_k} \hat{\psi}^{(-k)}_i + \frac{\sqrt{n}}{N} \sum_{k=1}^K\sum_{i \in \ci_k} \bar{\beta}^{\top}\left(\hat{\psi}^{(-k)}_i-\psi_i\right).
\end{align*}

First, we examine the behavior of $\sqrt{n}(\hat{\bar{\beta}}-\bar{\beta})$. We have that
\[
\hat{\bar{\beta}} = \left(\frac{1}{n} \sum_{i=1}^n \hat{\psi}_i\hat{\psi}_i^{\top}\right)^{-1} \left(\frac{1}{n} \sum_{i=1}^n D_i\hat{\psi}_i\right).
\]
and so
\[
\sqrt{n}(\hat{\bar{\beta}}-\bar{\beta}) = \underbrace{\left(\frac{1}{n}\sum_{i=1}^n \hat{\psi}_i\hat{\psi}_i^{\top}\right)^{-1}}_{\hat{A}_{n,N}} \times \underbrace{\left(\frac{1}{\sqrt{n}} \sum_{i=1}^n \left(D_i-\hat{\psi}_i^{\top}\bar{\beta}\right)\hat{\psi}_i\right)}_{\hat{b}_{n,N}}. 
\]
We can split the RHS into several terms:
\begin{align*}
& \sqrt{n}(\hat{\bar{\beta}}-\bar{\beta}) = \underbrace{\left(\frac{1}{n}\sum_{i=1}^n \psi_i\psi_i^{\top}\right)^{-1} }_{A_n}\times \underbrace{\left(\frac{1}{\sqrt{n}} \sum_{i=1}^n \left(D_i-\psi_i^{\top}\bar{\beta}\right)\psi_i\right)}_{b_n} + (\hat{A}_{n,N}-A_n)b_n + \hat{A}_{n,N}(\hat{b}_{n,N}-b_n).
\end{align*}
By the law of large numbers and the central limit theorem
\begin{align*}
A_n & = \left(\e_{\rct}[\psi(\Delta(X))\psi(\Delta(X))^{\top}]\right)^{-1} + o_p(1) \\
b_n & \stackrel{d}{\rightarrow} \mathcal{N}\left(0, \e_{\rct}[(D-\bar{\mu}(X))^2\psi(\Delta(X))\psi(\Delta(X))^{\top}]\right)
\end{align*}
(note that $\e_{\rct}[(D-\psi(\Delta(X))^{\top}\bar{\beta})\psi(\Delta(X))]=\e_{\rct}[(D-\bar{\mu}(X))\psi(\Delta(X))]=0$ by the first-order conditions defining $\bar{\mu}(\cdot)=\psi(\Delta(\cdot))^{\top}\bar{\beta}$).
We conclude
\[
\begin{split}
A_nb_n \stackrel{d}{\rightarrow} &\mathcal{N}\left(0, \left(\e_{\rct}[\psi(\Delta(X))\psi(\Delta(X))^{\top}]\right)^{-1}\e_{\rct}[(D-\bar{\mu}(X))^2\psi(\Delta(X))\psi(\Delta(X))^{\top}] \right. \\
& \qquad \left. \left(\e_{\rct}[\psi(\Delta(X))\psi(\Delta(X))^{\top}]\right)^{-1}\right).
\end{split}
\]
Next, we note that 
\[
\hat{A}_{n,N}^{-1}-A_n^{-1} = \frac{1}{n}\sum_{i=1}^n \hat{\psi}_i\left[\hat{\psi}_i-\psi_i\right]^{\top} + \frac{1}{n}\sum_{i=1}^n \left[\hat{\psi}_i-\psi_i\right]\psi_i^{\top}
\]
Letting $\psi_{i,j}$ denote the $j$-th component of $\psi_i$ (and similarly defining $\hat{\psi}_{i,j})$ as the $j$-th component of $\hat{\psi}_i$),
for each $1 \leq j,j' \leq p$ we have
\begin{align*}
\left|\frac{1}{n} \sum_{i=1}^n \psi_{i,j}\left[\hat{\psi}_{i,j'}-\psi_{i,j'}\right]\right| &  \leq \sqrt{\frac{1}{n} \sum_{i=1}^n \psi_{i,j}^2}\sqrt{\frac{1}{n} \sum_{i=1}^n \left[\hat{\psi}_{i,j'}-\psi_{i,j'}\right]^2} \\
& = O_p(1) \cdot o_p(n^{-1/2}) = o_p(1)
\end{align*}
where the inequality above follows by Cauchy-Schwarz,
and the $o_p(n^{-1/2})$ estimate follows from the rate condition~\eqref{eq:psi_rate} and Lemma 4 of~\citet{li2024efficient}
(which applies since $\hat{\Delta}(\cdot)$ is independent of all of the experimental observations).
By an identical argument
\begin{equation}
\label{eq:psi_CS}
     \left|\frac{1}{n} \sum_{i=1}^n \hat{\psi}_{i,j}\left[\hat{\psi}_{i,j'}-\psi_{i,j'} \right]\right| = o_p(n^{-1/2})=o_p(1)
\end{equation}
as well.
We conclude $\hat{A}_{n,N}^{-1}-A_n^{-1}=o_p(1)$ and hence $\hat{A}_{n,N}-A_n=o_p(1)$ and $(\hat{A}_{n,N}-A_n)b_n=o_p(1)$.
Finally, we write
\[
\hat{b}_{n,N}-b_n = \frac{1}{\sqrt{n}} \sum_{i=1}^n [\psi_i-\hat{\psi}_i]^{\top}\bar{\beta} \cdot \hat{\psi}_i +(D_i-\psi_i^{\top}\bar{\beta})(\hat{\psi}_i-\psi_i).
\]
By~\eqref{eq:psi_CS} we know that
\[
\frac{1}{\sqrt{n}} \sum_{i=1}^n \hat{\psi}_i [\psi_i-\hat{\psi}_i]^{\top}\bar{\beta} = o_p(1).
\]
Additionally
\[
d_{n,j} := \frac{1}{\sqrt{n}}\sum_{i=1}^n \left(D_i-\psi_i^{\top}\bar{\beta}\right)\left(\hat{\psi}_{i,j}-\psi_{i,j}\right)
\]
satisfies
\[
|d_{n,j}| \leq \sqrt{n} \cdot \sqrt{\frac{1}{n} \sum_{i=1}^n \left(D_i-\psi_i^{\top}\bar{\beta}\right)^2}\sqrt{\frac{1}{n}\sum_{i=1}^n \left(\hat{\psi}_{i,j}-\psi_{i,j}\right)^2} = o_p(1)
\]
for all $j=1,\ldots,p$ by Cauchy-Schwarz, the rate condition~\eqref{eq:psi_rate}, and Lemma 4 of~\citet{li2024efficient} once again.
Hence we can conclude $\hat{b}_{n,N}-b_n=o_p(1)$ and so 
\[
\hat{A}_{n,N}(\hat{b}_{n,N}-b_n)=(A_n+o_p(1))(\hat{b}_{n,N}-b_n)=o_p(1)
\]
as well. This shows that
\begin{equation}
\label{eq:beta_hat_expansion}
\begin{split}
\sqrt{n}(\hat{\bar{\beta}}-\bar{\beta}) & = A_nb_n + o_p(1) \\
& \stackrel{d}{\rightarrow} \mathcal{N}\left(0, \left(\e_{\rct}[\psi(\Delta(X))\psi(\Delta(X))^{\top}]\right)^{-1} \right. \\
& \qquad \left. \e_{\rct}[(D-\bar{\mu}(X))^2\psi(\Delta(X))\psi(\Delta(X))^{\top}]\left(\e_{\rct}[\psi(\Delta(X))\psi(\Delta(X))^{\top}]\right)^{-1}\right).
\end{split}
\end{equation}

Next, we note that for each $k=1,\ldots,K$ and $j=1,\ldots,p$, we have
\begin{align}
\frac{1}{|\ci_k|} \sum_{i \in \ci_k} \left|\hat{\psi}_{i,j}^{(-k)}- \psi_{i,j}\right| & \leq \sqrt{\frac{1}{|\ci_k|} \sum_{i \in \ci_k} \left(\hat{\psi}_{i,j}^{(-k)}- \psi_{i,j}\right)^2} \text{ by Jensen's inequality} \nonumber \\
& = o_p(n^{-1/2}) \text{ by Lemma 4 of~\citet{li2024efficient}} \label{eq:psi_hat_vs_psi}
\end{align}
and hence
\begin{align*}
\frac{1}{|\ci_k|} \sum_{i \in \ci_k} \bar{\beta}^{\top}\left(\hat{\psi}_{i,j}^{(-k)}- \psi_{i,j}\right)  = o_p(n^{-1/2})
\end{align*}
as well.
With $|\ci_k| = O(N)$ for each fold $k=1,\ldots,K$
we conclude that
\[
\frac{\sqrt{n}}{N} \sum_{k=1}^K \sum_{i \in \ci_k} \bar{\beta}^{\top}\left(\hat{\psi}_{i,j}^{(-k)}-\psi_{i,j}\right) = \sum_{k=1}^K \frac{|\ci_k|}{N} \cdot \frac{\sqrt{n}}{|\ci_k|} \sum_{i \in \ci_k} \bar{\beta}^{\top}\left(\hat{\psi}_{i,j}^{(-k)}-\psi_{i,j}\right) = o_p(1)
\]
and so in fact
\begin{align*}
\frac{\sqrt{n}}{N} \sum_{k=1}^K \sum_{i \in \ci_k}  \left(\hat{\bar{\beta}}^{\top}\hat{\psi}_i^{(-k)} - \bar{\beta}^{\top}\psi_i\right) & = \sqrt{n}(\hat{\bar{\beta}}-\bar{\beta})^{\top} \cdot \frac{1}{N} \sum_{k=1}^K \sum_{i \in \ci_k} \hat{\psi}^{(-k)}_i + o_p(1) \\
& \stackrel{d}{\rightarrow} \mathcal{N}(0,\Sigma)
\end{align*}
where the final convergence uses~\eqref{eq:beta_hat_expansion} along with the fact that 
\begin{equation}
\label{eq:avg_psi_hat_obs}
\frac{1}{N}\sum_{k=1}^K \sum_{i \in \ci_k} \hat{\psi}_i^{(-k)}=\e_{\obs}[\psi(\Delta(X))]+o_p(1)
\end{equation}
by the law of large numbers and~\eqref{eq:psi_hat_vs_psi}.

To prove asymptotic validity of the confidence interval,
it suffices to show that 
\[
n\hat{V}_{n,N} = \Sigma + o_p(1).
\]
To that end,
we first show that
\begin{equation}
\label{eq:ci_meat}
\frac{1}{n}\sum_{i=1}^n \hat{r}_i\hat{r}_i^{\top} = \e_{\rct}\left[(D-\bar{\mu}(X))^2\psi(\Delta(X))\psi(\Delta(X))^{\top}\right] + o_p(1).
\end{equation}
Since
\[
\frac{1}{n} \sum_{i=1}^n (D_i-\bar{\mu}(X_i))^2\psi_i\psi_i^{\top} =  \e_{\rct}\left[(D-\bar{\mu}(X))^2\psi(\Delta(X))\psi(\Delta(X))^{\top}\right] + o_p(1)
\]
by the law of large numbers,
it suffices to show that
\begin{equation}
\label{eq:ri_hat}
\frac{1}{n} \sum_{i=1}^n \hat{r}_i\hat{r}_i^{\top} - (D_i-\bar{\mu}(X_i))^2\psi_i\psi_i^{\top} = o_p(1).
\end{equation}
To do so, we fix $1 \leq j,j' \leq p$ and write
\begin{align*}
& \frac{1}{n} \sum_{i=1}^n \left(D_i-\hat{\bar{\mu}}(X_i)\right)^2\hat{\psi}_{i,j}\hat{\psi}_{i,j'} - (D_i-\bar{\mu}(X_i))^2\psi_{i,j}\psi_{i,j'} = t_n+u_n
\end{align*}
where
\begin{align*}
t_n & = \frac{1}{n} \sum_{i=1}^n \left(2D_i-\bar{\mu}(X_i)-\hat{\bar{\mu}}(X_i)\right)\left(\bar{\mu}(X_i)-\hat{\bar{\mu}}(X_i)\right)\psi_{i,j}\psi_{i,j'} \\
& = 2 \cdot \frac{1}{n} \sum_{i=1}^n \left(D_i-\bar{\mu}(X_i)\right)\left(\bar{\mu}(X_i)-\hat{\bar{\mu}}(X_i)\right)\psi_{i,j}\psi_{i,j'} + \frac{1}{n}\sum_{i=1}^n \left(\bar{\mu}(X_i)-\hat{\bar{\mu}}(X_i)\right)^2\psi_{i,j}\psi_{i,j'} \\
u_n & = \frac{1}{n} \sum_{i=1}^n \left(D_i-\hat{\bar{\mu}}(X_i)\right)^2\left(\hat{\psi}_{i,j}\hat{\psi}_{i,j'}-\psi_{i,j}\psi_{i,j'}\right)
\end{align*}
and show that $t_n+u_n=o_p(1)$.

To bound the first term of $t_n$, we apply Cauchy-Schwarz and then H\"older's inequality:
\begin{align*}
& \left| \frac{1}{n} \sum_{i=1}^n \left(D_i-\bar{\mu}(X_i)\right)\left(\bar{\mu}(X_i)-\hat{\bar{\mu}}(X_i)\right)\psi_{i,j}\psi_{i,j'} \right| \\
& \quad \leq \left(\frac{1}{n}\sum_{i=1}^n (D_i-\bar{\mu}(X_i))^2\right)^{\frac{1}{2}} \times \left(\frac{1}{n} \sum_{i=1}^n \left(\bar{\mu}(X_i)-\hat{\bar{\mu}}(X_i)\right)^2\psi_{i,j}^2\psi_{i,j'}^2\right)^{\frac{1}{2}} \\
& \quad \leq  \left(\frac{1}{n}\sum_{i=1}^n (D_i-\bar{\mu}(X_i))^2\right)^{\frac{1}{2}} \times \left(\frac{1}{n}\sum_{i=1}^n \left|\hat{\bar{\mu}}(X_i)-\bar{\mu}(X_i)\right|^{2+\epsilon}\right)^{\frac{1}{2+\epsilon}} \times \left(\frac{1}{n} \sum_{i=1}^n |\psi_{i,j}|^{2q}||\psi_{i,j'}|^{2q}\right)^{\frac{1}{2q}}
\end{align*}
where $q:=1+2/\epsilon$ satisfies
\[
\frac{2}{2+\epsilon} + \frac{1}{q}=1.
\]
Evidently
\[
\frac{1}{n} \sum_{i=1}^n (D_i-\bar{\mu}(X_i))^2 = O_p(1)
\]
since $D_i$ and $\bar{\mu}(X_i)$ are both square integrable while
\begin{equation}
\label{eq:all_moments}
\frac{1}{n} \sum_{i=1}^n |\psi_{i,j}|^{k}||\psi_{i,j'}|^{k} = O_p(1), \quad \forall k < \infty
\end{equation}
by the assumption that $\psi_j(\Delta(X))$ and $\psi_{j'}(\Delta(X))$ have finite moments of all orders in the RCT.
Thus,
\begin{align*}
& \left| \frac{1}{n} \sum_{i=1}^n \left(D_i-\bar{\mu}(X_i)\right)\left(\bar{\mu}(X_i)-\hat{\bar{\mu}}(X_i)\right)\psi_{i,j}\psi_{i,j'} \right| = O_p(1) \times \left(\frac{1}{n}\sum_{i=1}^n \left|\hat{\bar{\mu}}(X_i)-\bar{\mu}(X_i)\right|^{2+\epsilon}\right)^{\frac{1}{2+\epsilon}}
\end{align*}
We bound
\begin{align*}
\left(\frac{1}{n} \sum_{i=1}^n \left|\hat{\psi}_i^{\top}\hat{\bar{\beta}}-\psi_i^{\top}\bar{\beta}\right|^{2+\epsilon}\right)^{\frac{1}{2+\epsilon}} & \leq \left(\frac{1}{n} \sum_{i=1}^n \left|(\hat{\psi}_i-\psi_i)^{\top}\hat{\bar{\beta}}\right|^{2+\epsilon}\right)^{\frac{1}{2+\epsilon}} + \left(\frac{1}{n} \sum_{i=1}^n \left|(\hat{\bar{\beta}}-\bar{\beta})^{\top}\psi_i\right|^{2+\epsilon}\right)^{\frac{1}{2+\epsilon}} \\
& \leq \sum_{\tilde{j}=1}^p \left(\frac{1}{n} \sum_{i=1}^n \left|(\hat{\psi}_{i,\tilde{j}}-\psi_{i,\tilde{j}})\hat{\bar{\beta}}_{\tilde{j}}\right|^{2+\epsilon}\right)^{\frac{1}{2+\epsilon}} + \left(\frac{1}{n} \sum_{i=1}^n \left|(\hat{\bar{\beta}}_{\tilde{j}}-\bar{\beta}_{\tilde{j}})\psi_{i,\tilde{j}}\right|^{2+\epsilon}\right)^{\frac{1}{2+\epsilon}}
\end{align*}
by Minkowski's inequality.
Evidently for each $\tilde{j}=1,\ldots,p$ we have
\begin{align*}
\frac{1}{n} \sum_{i=1}^n \left|(\hat{\psi}_{i,\tilde{j}}-\psi_{i,\tilde{j}})\hat{\bar{\beta}}_{\tilde{j}}\right|^{2+\epsilon} = |\hat{\bar{\beta}}_{\tilde{j}}|^{2+\epsilon} \cdot \frac{1}{n} \sum_{i=1}^n \left|\hat{\psi}_{i,\tilde{j}}-\psi_{i,\tilde{j}}\right|^{2+\epsilon} = O_p(1) \cdot o_p(1) = o_p(1)
\end{align*}
since $\|\psi_{\tilde{j}} \circ \hat{\Delta}-\psi_{\tilde{j}} \circ \Delta\|_{2+\epsilon,\rct} = o_p(1)$,
and similarly 
\[
\frac{1}{n} \sum_{i=1}^n \left|(\hat{\bar{\beta}}_{\tilde{j}}-\bar{\beta}_{\tilde{j}})\psi_{i,\tilde{j}}\right|^{2+\epsilon} = \left|\hat{\bar{\beta}}_{\tilde{j}}-\bar{\beta}_{\tilde{j}}\right|^{2+\epsilon} \cdot \frac{1}{n} \sum_{i=1}^n \left|\psi_{i,\tilde{j}}\right|^{2+\epsilon} = O_p(n^{-\frac{1}{2} \cdot (2+\epsilon)}) \cdot O_p(1) = o_p(1)
\]
by~\eqref{eq:beta_hat_expansion}.
We conclude that 
\begin{equation}
\label{eq:mu_hat_vs_mu_rct}
\left(\frac{1}{n} \sum_{i=1}^n \left|\hat{\bar{\mu}}(X_i)-\bar{\mu}(X_i)\right|^{2+\epsilon}\right)^{\frac{1}{2+\epsilon}} = \left(\frac{1}{n} \sum_{i=1}^n \left|\hat{\psi}_i^{\top}\hat{\bar{\beta}}-\psi_i^{\top}\bar{\beta}\right|^{2+\epsilon}\right)^{\frac{1}{2+\epsilon}} = o_p(1)
\end{equation}
and so in fact
\[
\left| \frac{1}{n} \sum_{i=1}^n \left(D_i-\bar{\mu}(X_i)\right)\left(\bar{\mu}(X_i)-\hat{\bar{\mu}}(X_i)\right)\psi_{i,\tilde{j}}\psi_{i,\tilde{j}'} \right| = O_p(1) \times o_p(1) = o_p(1).
\]
For the second term of $t_n$, we have
\begin{align*}
\left|\frac{1}{n} \sum_{i=1}^n (\bar{\mu}(X_i)-\hat{\bar{\mu}}(X_i))^2 \psi_{i,j}\psi_{i,j'}\right| & \leq \left(\frac{1}{n} \sum_{i=1}^n \left|\hat{\bar{\mu}}(X_i)-\bar{\mu}(X_i)\right|^{2+\epsilon}\right)^{\frac{2}{2+\epsilon}} \left(\frac{1}{n}\sum_{i=1}^n |\psi_{i,j}|^{q}|\psi_{i,j'}|^q\right)^{\frac{1}{q}} \\
& = o_p(1) \cdot O_p(1) = o_p(1)
\end{align*}
by H\"older's inequality,~\eqref{eq:all_moments}, and~\eqref{eq:mu_hat_vs_mu_rct}.
Thus, we have shown $t_n=o_p(1)$.

To show $u_n=o_p(1)$, by Cauchy-Schwarz
\begin{align*}
|u_n| & \leq  \left(\frac{1}{n} \sum_{i=1}^n |D_i-\hat{\bar{\mu}}(X_i)|^4\right)^{1/2} \left(\frac{1}{n} \sum_{i=1}^n \left|\hat{\psi}_{i,j}\hat{\psi}_{i,j'}-\psi_{i,j}\psi_{i,j'}\right|^2\right)^{1/2}.
\end{align*}
By Minkowski
\begin{align*}
\left(\frac{1}{n} \sum_{i=1}^n |D_i-\hat{\bar{\mu}}(X_i)|^4\right)^{1/4} & \leq \left(\frac{1}{n} \sum_{i=1}^n D_i^4\right)^{1/4} + \left(\frac{1}{n} \sum_{i=1}^n \left|\psi(\hat{\Delta}(X_i))^{\top}\hat{\bar{\beta}}\right|^4\right)^{1/4}  \\
& \leq O_p(1) + \|\hat{\bar{\beta}}\|_2 \left(\frac{1}{n} \sum_{i=1}^n \left\|\psi(\hat{\Delta}(X_i))\right\|_2^4\right)^{1/4} \\
& = O_p(1)
\end{align*}
using the fact that $\e_{\rct}[|D|^4]<\infty$ and that~\eqref{eq:all_orders} holds with $k=4$. 
Another application of Minkowski shows that
\begin{align*}
\left(\frac{1}{n} \sum_{i=1}^n \left|\hat{\psi}_{i,j}\hat{\psi}_{i,j'}-\psi_{i,j}\psi_{i,j'}\right|^2\right)^{1/2} & \leq \left(\frac{1}{n} \sum_{i=1}^n \left|(\hat{\psi}_{i,j}-\psi_{i,j})\hat{\psi}_{i,j'}\right|^2\right)^{1/2} + \left(\frac{1}{n} \sum_{i=1}^n \left|\psi_{i,j}(\hat{\psi}_{i,j'}-\psi_{i,j'})\right|^2\right)^{1/2}
\end{align*}
Then we use H\"older's inequality to bound
\begin{align*}
\frac{1}{n} \sum_{i=1}^n \left|\psi_{i,j}(\hat{\psi}_{i,j'}-\psi_{i,j'})\right|^2 & \leq \left(\frac{1}{n}\sum_{i=1}^n \left|\hat{\psi}_{i,j'}-\psi_{i,j'}\right|^{2+\epsilon}\right)^{\frac{2}{2+\epsilon}}\left(\frac{1}{n} \sum_{i=1}^n \left|\psi_{i,j}\right|^{2q}\right)^{\frac{1}{q}} \\
& = o_p(1) \cdot O_p(1) = o_p(1)
\end{align*}
where the rate estimates hold by the assumptions that 
\[
\|\psi_{j'} \circ \hat{\Delta} - \psi_{j'} \circ \Delta\|_{2+\epsilon,\rct} = o_p(1)
\]
and $\e_{\rct}[\|\psi(\Delta(X))\|^{2q}]< \infty$.
By an identical argument,
except using~\eqref{eq:all_orders} with $k=2q$ in place of the assumption $\e_{\rct}[\|\psi(\Delta(X))\|^{2q}]< \infty$,
we have
\begin{align*}
\frac{1}{n} \sum_{i=1}^n \left|(\hat{\psi}_{i,j}-\psi_{i,j})\hat{\psi}_{i,j'}\right|^2 & \leq \left(\frac{1}{n}\sum_{i=1}^n \left|\hat{\psi}_{i,j}-\psi_{i,j}\right|^{2+\epsilon}\right)^{\frac{2}{2+\epsilon}}\left(\frac{1}{n} \sum_{i=1}^n \left|\hat{\psi}_{i,j'}\right|^{2q}\right)^{\frac{1}{q}} \\
& = o_p(1) \cdot O_p(1) = o_p(1).
\end{align*}
We conclude that in fact $u_n=o_p(1)$,
and so we have shown~\eqref{eq:ri_hat}
and thus~\eqref{eq:ci_meat}.

Now we write
\[
n\hat{V}_{n,N} = \hat{a}_{n,N}^{\top} \left(\frac{1}{n} \sum_{i=1}^n \hat{r}_i\hat{r}_i^{\top}\right)\hat{a}_{n,N} + \frac{n}{N} \cdot \frac{1}{N} \sum_{k=1}^K \sum_{i \in \ci_k} \left(\hat{\bar{\mu}}^{(-k)}(X_i)-\hat{\bar{\tau}}\right)^2
\]
to see that in light of~\eqref{eq:ci_meat},
to conclude the proof it suffices to show that
\begin{align*}
\hat{a}_{n,N} = \bar{\alpha} + o_p(1) \qquad \text{and} \qquad \frac{1}{N} \sum_{k=1}^K \sum_{i \in \ci_k} \left(\hat{\bar{\mu}}^{(-k)}(X_i)-\hat{\bar{\tau}}\right)^2 = O_p(1).
\end{align*}
But the statement that $\hat{a}_{n,N} = \bar{\alpha}+o_p(1)$ follows immediately from~\eqref{eq:avg_psi_hat_obs} and the fact that 
\[
\hat{A}_{n,N}=A_n+o_p(1)=\left(\e_{\rct}\left[\psi(\Delta(X))\psi(\Delta(X))^{\top}\right]\right)^{-1} + o_p(1)
\]
as shown in the first part of this proof.
Finally,
the assertion
\[
\frac{1}{N} \sum_{k=1}^K \sum_{i \in \ci_k} \left(\hat{\bar{\mu}}^{(-k)}(X_i)-\hat{\bar{\tau}}\right)^2 = O_p(1)
\]
is apparent since
\[
\frac{1}{|\ci_k|} \sum_{i \in \ci_k} \left(\hat{\mu}^{(-k)}(X_i)\right)^2
\]
has finite expectation for all $k=1,\ldots,K$ by assumption.

\subsection{Proof of Theorem~\ref{thm:efficiency}}

In the setting of Theorem~\ref{thm:efficiency},
there are no structural restrictions on the distribution of $W_i=(Q_i,X_i,Y_i)$,
i.e., we have a fully nonparametric model.
Then the corresponding semiparametric tangent space is equal to $\ch_W^0$,
the set of all mean zero, square-integrable functions of $W=(Q,X,Y)$ under the true data generating distribution $P^*$~\citep{bickel1993efficient, kennedy2024semiparametric}.
Hence, for the proof of the efficiency bound
it suffices to show that
\begin{equation}
\label{eq:pathwise_derivative}
\frac{\partial \bar{\tau}(0)}{\partial \theta} = \e\left[\psi(W;\bar{\tau},\eta)s(W;0)\right]
\end{equation}
for all smooth parametric submodels $\{P_{\theta} \mid \theta \in \Theta\}$ with corresponding score $s(w;\theta)$,
where for each $\theta \in \Theta$,
$\bar{\tau}(\theta)$ denotes the value of the parameter $\bar{\tau}$ under the distribution $P_{\theta}$.
Following standard texts on semiparametric theory~\citep{bickel1993efficient, van2000asymptotic},
a parametric submodel $\{P_{\theta} \mid \theta \in \Theta\}$
is a collection of distributions $P_{\theta}$ indexed by $\theta$ in an open interval $\Theta$ containing 0 such that $P_0=P^*$,
the true data-generating process.
Above and throughout this proof, for brevity we use the notation
\[
\frac{\partial F(a)}{\partial \theta} \equiv \frac{\partial}{\partial \theta} F(\theta) \Big|_{\theta=a}.
\]
for any function $F$.

We fix an arbitrary parametric submodel, which contains densities of the form
\[
p(w;\theta) = f(q;\theta)g(x \mid q;\theta) h(y \mid x,q;\theta)
\]
where $w=(q,x,y)$ and $f$, $g$, and $h$ are the density of $Q$, the conditional density of $X$ given $Q$, and the conditional density of $Y$ given $(X,Q)$, respectively with respect to measures $\lambda_Q$, $\lambda_{X \mid Q}$, and $\lambda_{Y \mid QX}$, respectively.
Then
\begin{align*}
\bar{\tau}(\theta) & = \left[\frac{\int \indic(q \neq 2)\psi(\Delta(x;\theta))f(q;\theta)g(x \mid q;\theta) d(\lambda_Q \times \lambda_{X \mid Q})(q,x)}{1-f(2;\theta)}\right]^{\top}\bar{\beta}(\theta), \quad \\
\Delta(x;\theta) & = \int y \left[h(y \mid x,1;\theta)-h(y \mid x,0;\theta)\right]  d\lambda_{Y \mid QX}(y), \quad \\
\bar{\beta}(\theta) & = \left[\int g(x \mid 2;\theta) \psi(\Delta(x;\theta))\psi(\Delta(x;\theta))^{\top}  d\lambda_{X \mid Q}(x)\right]^{-1}  \times \\
& \left[\int g(x \mid 2;\theta)h(y \mid x,2;\theta)y\psi(\Delta(x;\theta))d(\lambda_{X \mid Q} \times \lambda_{Y \mid QX})(x,y)  \right].
\end{align*}
It follows by the product rule that
\begin{align*}
\frac{\partial \bar{\tau}(0)}{\partial \theta} & = \left(\e[\psi(\Delta(X)) \mid Q \neq 2]\right)^{\top} \frac{\partial \bar{\beta}(0)}{\partial \theta} \\
& + \left[\frac{\partial}{\partial \theta} \left(\frac{\int \indic(q \neq 2)f(q;0)g(x \mid q;0)\psi(\Delta(x;0)) d(\lambda_Q \times \lambda_{X \mid Q})(q,x)}{1-f(2;0)}\right) \right]^{\top} \bar{\beta}
\end{align*}
where above and hereafter in this proof, all expectations are with respect to the true data-generating distribution $P_0$.

Differentiating under the integral sign, we get
\begin{align*}
\frac{\partial \bar{\beta}(0)}{\partial \theta} & = \left(\e\left[\psi(\Delta(X))\psi(\Delta(X))^{\top} \mid Q=2\right]\right)^{-1}\left(B_1+B_2-B_3-B_4\right)
\end{align*}
for
\begin{align*}
B_1 & = \e\left[YD_{gh}(X,Y \mid 2)\psi(\Delta(X)) \Big| Q=2\right] \\
B_2 & = \e\left[Y\frac{\partial \Delta(X;0)}{\partial \theta} \dot{\psi}(\Delta(X)) \Big| Q=2\right] \\
B_3 & = \e\left[\frac{\partial \log g(X \mid 2;0)}{\partial \theta}  \psi(\Delta(X))\psi(\Delta(X))^{\top}\bar{\beta} \Big| Q=2\right] \\
B_4 & = \e\left[\frac{\partial \Delta(X;0)}{\partial \theta} \left(\psi(\Delta(X))\dot{\psi}(\Delta(X))^{\top} + \dot{\psi}(\Delta(X)) \psi(\Delta(X))^{\top}\right)\bar{\beta} \Big| Q=2\right]
\end{align*}
where we let
\[
D_{gh}(X,Y \mid q) = \frac{ \partial \log g(X \mid q;0)}{\partial \theta} + \frac{\partial \log h(Y \mid X,q;0)}{\partial \theta}, \quad q=0,1,2.
\]
On the other hand,
\[
\frac{\partial}{\partial \theta} \left(\frac{\int \indic(q \neq 2)f(q;0)g(x \mid q;0)\psi(\Delta(x;0)) d(\lambda_Q \times \lambda_{X \mid Q})(q,x)}{1-f(2;0)}\right) = A_1+A_2+A_3
\]
for
\begin{align*}
A_1 & = \e\left[\left(\frac{\partial \log f(Q;0)}{\partial \theta} + \frac{\partial \log g(X \mid Q;0)}{\partial \theta}\right)\psi(\Delta(X)) \Big| Q \neq 2\right] \\
A_2 & = \e\left[\frac{\partial \Delta(X;0)}{\partial \theta} \dot{\psi}(\Delta(X)) \Big| Q \neq 2\right] \\
A_3 & = \frac{\rho_2}{1-\rho_2} \frac{\partial \log f(2;0)}{\partial \theta} \e\left[\psi(\Delta(X)) \Big| Q \neq 2\right].
\end{align*}
Thus, to show~\eqref{eq:pathwise_derivative} it suffices to prove that
\[
\e[\psi(W;\bar{\tau},\eta)s(W;0)] = \bar{\alpha}^{\top}(B_1+B_2-B_3-B_4) + (A_1+A_2+A_3)^{\top}\bar{\beta}.
\]
where
\[
s(w;0) = \frac{\partial \log f(q;0)}{\partial \theta} + D_{gh}(x,y \mid q).
\]
is the score function corresponding to the parametric submodel at $\theta=0$.

To that end, we write $\psi(w;\bar{\tau},\eta) = \psi_1(w;\bar{\tau},\eta)+\psi_2(w;\bar{\tau},\eta)$ where
\begin{align*}
\psi_1(w;\bar{\tau},\eta) & = \frac{\indic(q \neq 2)}{1-\rho_2}\left(\bar{\mu}(x)-\bar{\tau}+\iota(w;\eta)\kappa(x;\eta)\right), \quad \text{and} \\
\psi_2(w;\bar{\tau},\eta) & = \frac{\indic(q = 2)}{\rho_2}(y-\bar{\mu}(x))\bar{\alpha}^{\top}\psi(\Delta(x)).
\end{align*}
The first-order conditions for $\bar{\mu}$ imply that
\[
\e[\psi_2(W;\bar{\tau},\eta)] = 0
\]
and so
\[
\e\left[\psi_2(W;\bar{\tau},\eta) \frac{\partial \log f(q;0)}{\partial \theta}\right] = \e\left[\psi_2(W;\bar{\tau},\eta) \frac{\partial \log f(2;0)}{\partial \theta}\right] = \frac{\partial \log f(2;0)}{\partial \theta} \e\left[\psi_2(W;\bar{\tau},\eta) \right]=0.
\]
Thus
\begin{align*}
& \e\left[\psi_2(W;\bar{\tau},\eta)s(W;0)\right] \\
& \quad = \e\left[\frac{\indic(Q = 2)}{\rho_2}(Y-\bar{\mu}(X))\bar{\alpha}^{\top}\psi(\Delta(X))D_{gh}(X,Y \mid 2)\right] \\
& \quad = \bar{\alpha}^{\top}\left(\e\left[YD_{gh}(X,Y \mid 2)\psi(\Delta(X)) \mid Q=2\right]-\e\left[\bar{\mu}(X)D_{gh}(X,Y \mid 2)\psi(\Delta(X)) \mid   Q=2\right]\right) \\
& \quad = \bar{\alpha}^{\top}(B_1-B_3)
\end{align*}
and the final equality uses the fact that
\begin{equation}
\label{eq:h_0}
\e\left[F(X,Q)\frac{\partial \log h(Y \mid X,Q;0)}{\partial \theta}\right] = 0
\end{equation}
for any square integrable function $F$,
applied to $F(x,q)=\indic(q=2)\bar{\mu}(x)\psi(\Delta(x))$.

Next, we have
\[
\e[\psi_1(W;\bar{\tau},\eta)s(W;0)]  = \e\left[\frac{\indic(Q \neq 2)}{1-\rho_2}(\bar{\mu}(X)-\bar{\tau})D_{fg}(Q,X)\right] + \e\left[\frac{\iota(W;\eta)\kappa(X;\eta)}{1-\rho_2}\frac{\partial \log h(Y \mid X,Q;0)}{\partial \theta}\right] \\
\]
where some terms cancel by~\eqref{eq:h_0} and observing by conditioning on $(X,Q)$ that both
\begin{align*}
\e\left[\frac{\iota(W;\eta)\kappa(X;\eta)}{1-\rho_2}\frac{\partial \log f(Q;0)}{\partial \theta}\right] & = \sum_{q=0}^1 \frac{\partial \log f(q;0)}{\partial \theta}\e\left[\frac{\kappa(X;\eta)\indic(Q=q)}{(r(X)-(1-q))(1-\rho_2)}(Y-m_q(X))\right]  \\
& = 0+0 = 0
\end{align*}
and
\begin{align*}
\e\left[\frac{\iota(W;\eta)\kappa(X;\eta)}{1-\rho_2} \frac{\partial \log g(X \mid Q;0)}{\partial \theta}\right] & = \sum_{q=0}^1 \e\left[\frac{\kappa(X;\eta)\indic(Q=q)}{(r(X)-(1-q))(1-\rho_2)}\frac{\partial \log g(X \mid Q;0)}{\partial \theta}(Y-m_q(X))\right]  \\
& = 0+0 = 0.
\end{align*}
We further compute
\begin{align*}
& \e\left[\frac{\indic(Q \neq 2)}{1-\rho_2}(\bar{\mu}(X)-\bar{\tau})D_{fg}(Q,X)\right] \\
& \quad = \left(\e[D_{fg}(Q,X)\psi(\Delta(X)) \mid Q \neq 2]\right)^{\top}\bar{\beta} - \e\left[\frac{\partial \log f(Q;0)}{\partial \theta} \Big| Q \neq 2\right]\left(\e[\psi(\Delta(X)) \mid Q \neq 2]\right)^{\top}\bar{\beta} \\
& \quad = \left(\e[D_{fg}(Q,X)\psi(\Delta(X)) \mid Q \neq 2]\right)^{\top}\bar{\beta} + \frac{\rho_2}{1-\rho_2} \frac{\partial \log f(2;0)}{\partial \theta} \left(\e[\psi(\Delta(X)) \mid Q \neq 2]\right)^{\top}\bar{\beta} \\
& \quad = (A_1+A_3)^{\top}\bar{\beta} 
\end{align*}
where the second equality follows by the law of total expectation since
\begin{align*}
0 = \e\left[\frac{\partial \log f(Q;0)}{\partial \theta}\right] = (1-\rho_2) \e\left[\frac{\partial \log f(Q;0)}{\partial \theta} \Big| Q \neq 2\right] + \rho_2 \frac{\partial}{\partial \theta} \log f(2;0).
\end{align*}
Finally,
\begin{align*}
& \e\left[\frac{\iota(W;\eta)\kappa(X;\eta)}{1-\rho_2}\frac{\partial \log h(Y \mid X,Q;0)}{\partial \theta}\right] \\
& \quad = \e\left[\frac{Y}{1-\rho_2}\left(\frac{\indic(Q=1)}{r(X)}-\frac{\indic(Q=0)}{1-r(X)}\right)\kappa(X;\eta)\frac{\partial \log h(Y \mid X,Q;0)}{\partial \theta}\right] \\
& \quad = \e\left[(1-\rho_2)^{-1}\left(\frac{\indic(Q=1)}{r(X)}-\frac{\indic(Q=0)}{1-r(X)}\right)\kappa(X;\eta)\e\left[Y\frac{\partial \log h(Y \mid X,Q,0)}{\partial \theta} \Big| X,Q\right]\right] \\
& \quad = (1-\rho_2)^{-1} \sum_{q=0}^1 \e\left[\frac{\Pr(Q=q \mid X)}{r(X)-(1-q)}\kappa(X;\eta) \frac{\partial \int y h(y \mid X,q;0)d\lambda_{Y \mid QX}(y)}{\partial \theta}\right] \\
& \quad = (1-\rho_2)^{-1}\e\left[\Pr(Q \neq 2 \mid X)\kappa(X;\eta)\frac{\partial \Delta(X;0)}{\partial \theta}\right] \\
& \quad = \e\left[\frac{\partial \Delta(X;0)}{\partial \theta}\left(\dot{\psi}(\Delta(X))^{\top}\bar{\beta} + \Lambda(X)m_2(X)\bar{\alpha}^{\top}\dot{\psi}(\Delta(X))\right) \Big| Q \neq 2\right] \\
& \quad - \e\left[\frac{\partial \Delta(X;0)}{\partial \theta}\Lambda(X)\bar{\alpha}^{\top}\left[\psi(\Delta(X))\dot{\psi}(\Delta(X))^{\top} + \dot{\psi}(\Delta(X))\psi(\Delta(X))^{\top}\right]\bar{\beta} \Big| Q \neq 2 \right] \\
& \quad = A_2^{\top}\bar{\beta} + \bar{\alpha}^{\top}(B_2-B_4)
\end{align*}
where the first equality follows from~\eqref{eq:h_0} applied to the functions $F(x,q)=\frac{m_r(x)\indic(Q=r)\kappa(x;\eta)}{(r(x)-q)(1-\rho_2)}$ for $r=0,1$,
and the fourth equality follows from observing that
\[
\Pr(Q=q \mid X) = \Pr(Q=q\mid X,Q \neq 2)P(Q \neq 2 \mid X) = (r(X)-(1-q))P(Q \neq 2 \mid X)(-1)^{q+1}.
\]
We conclude that indeed
\begin{align*}
\e[\psi(W;\bar{\tau},\eta)s(W;0)] & = \e[\psi_1(W;\bar{\tau},\eta)s(W,0)] + \e[\psi_2(W;\bar{\tau},\eta)s(W,0)] \\
& = \bar{\alpha}^{\top}(B_1+B_2-B_3-B_4) + (A_1+A_2+A_3)^{\top}\bar{\beta}
\end{align*} 
which concludes the proof of~\eqref{eq:pathwise_derivative}.

It remains to show that whenever the efficiency bound $V_{\eff}$ is finite,
$\rho_2 V_{\eff} \rightarrow \Sigma$ as $M \rightarrow \infty$ with $\rho_2 \rightarrow 0$ while the conditional distribution of $W=(Q,X,Y)$ given $\indic(Q = 2)$ remains constant.
To see this,
we write
\begin{align*}
V_{\eff} & = (1-\rho_2)^{-1} \e\left[(\bar{\mu}(X)-\bar{\tau}\iota(W;\eta)\kappa(X;\eta))^2 \mid Q \neq 2\right]  \\
& \quad + \rho_2^{-1} \e\left[(Y-\bar{\mu}(X))^2\bar{\alpha}^{\top}\psi(\Delta(X))\psi(\Delta(X))^{\top}\bar{\alpha} \mid Q=2\right]
\end{align*}
The assumption $V_{\eff}<\infty$ implies $0<\rho_2<1$ and both conditional expectations
\begin{align*}
E_1 & \equiv \e\left[(\bar{\mu}(X)-\bar{\tau}\iota(W;\eta)\kappa(X;\eta))^2 \mid Q \neq 2\right] \quad \text{ and } \\
E_2 & \equiv \e\left[(Y-\bar{\mu}(X))^2\bar{\alpha}^{\top}\psi(\Delta(X))\psi(\Delta(X))^{\top}\bar{\alpha} \mid Q=2\right]
\end{align*}
are finite.
These conditional expectations $E_1$ and $E_2$ remain fixed in our asymptotics.
Then as desired, evidently
\[
\rho_2 V_{\eff} = \frac{\rho_2}{1-\rho_2}E_1 + E_2 \rightarrow E_2 = \Sigma
\]
as $M \rightarrow \infty$ and $\rho_2 \rightarrow 0$.

\section{Efficiency bound for \texorpdfstring{$\bar{\tau}$}{our estimand} when \texorpdfstring{$\Delta(\cdot)$}{the treatment-control contrast} is known}
\label{app:eff_bound}

Here we present the semiparametric efficiency bound 
(under the nested formulation of our setting introduced in Section~\ref{sec:discussion}) for our estimand $\bar{\tau}$ when the treatment-control contrast $\Delta(\cdot)$ is known.
As discussed in Section~\ref{sec:discussion},
prior work on efficient estimation of projection estimands gives the semiparametric efficiency bound for the estimand $\bar{\theta}=(1-\rho_2)\bar{\tau}=\e[\indic(Q \neq 2)\bar{\tau}]$ when $\Delta(\cdot)$ is known as a special case.
In particular,
from Proposition 4 of~\citet{newey1994asymptotic},
we see that the efficient influence function (EIF) for $\bar{\theta}$, evaluated at a generic observation $w=(q,x,y)$, is
\[
\psi_{\bar{\theta}}(w) = \indic(q \neq 2)\bar{\mu}(x) -\bar{\theta} + \frac{\indic(q=2)(1-\rho_2)\bar{\alpha}^{\top}\psi(\Delta(x)))(y-\bar{\mu}(x))}{\rho_2}
\]
when $\Delta(\cdot)$ is known,
where $\bar{\alpha}$ is given by~\eqref{eq:alpha_nested}.
Then to get the EIF $\psi_{\bar{\tau}}(\cdot)$ of $\bar{\tau}=(1-\rho_2)^{-1}\bar{\theta}$ with $\Delta(\cdot)$ known,
we recall that the EIF for $1-\rho_2 = \e[\indic(Q \neq 2)]$ is simply
\[
\psi_{1-\rho_2}(w) = \indic(q \neq 2)-(1-\rho_2) = \rho_2-\indic(q=2).
\]
Then using the properties of influence functions as derivatives (e.g., Theorems 20.8 and 20.9 of~\citet{van2000asymptotic} or~\citet{kennedy2024semiparametric}),
we can apply the quotient rule to conclude
\begin{equation}
\label{eq:eif_Delta_known}
\psi_{\bar{\tau}}(w) = \frac{(1-\rho_2)\psi_{\bar{\theta}}(w)-\bar{\theta}\psi_{1-\rho_2}(w)}{(1-\rho_2)^2} = \frac{\indic(q \neq 2)}{1-\rho_2}(\bar{\mu}(x)-\bar{\tau}) + \frac{\indic(q=2)}{\rho_2}(y-\bar{\mu}(x))\bar{\alpha}^{\top}\psi(\Delta(x)).
\end{equation}
Note that this differs from the EIF in the case $\Delta$ must be estimated (Theorem~\ref{thm:efficiency}) by the omission of the term
\[
\frac{\indic(q \neq 2)}{1-\rho_2}\iota(w;\eta)\kappa(x;\eta).
\]

From~\eqref{eq:eif_Delta_known},
the efficiency bound (when $\Delta$ is known) is
\[
\e[\psi_{\bar{\tau}}(W)^2] = \frac{\e[(\bar{\mu}(X)-\bar{\tau})^2 \mid Q \neq 2]}{(1-\rho_2)} + \frac{\e[(Y-\bar{\mu}(X))^2\bar{\alpha}^{\top}\psi(\Delta(X))\psi(\Delta(X))^{\top}\bar{\alpha} \mid Q=2]}{\rho_2}.
\]
Assuming this is finite,
we observe that as $M \rightarrow \infty$ with $\rho_2 \rightarrow 0$ and the distribution of $(Q,X,Y)$ given $\indic(Q = 2)$ remaining constant, we have
\[
\rho_2\e[\psi_{\bar{\tau}}(W)^2] \rightarrow \e[(Y-\bar{\mu}(X))^2\bar{\alpha}^{\top}\psi(\Delta(X))\psi(\Delta(X))^{\top}\bar{\alpha} \mid Q=2] = \Sigma.
\]
In conjunction with Theorem~\ref{thm:efficiency}, this shows the claim at the end of Section~\ref{sec:discussion} that in this asymptotic regime,
the need to estimate $\Delta(\cdot)$ does not affect the efficiency bound.

\section{Nuisance function estimation for numerical studies}
\label{app:nuisance_estimation}

For the simulation study in Section~\ref{sec:1d_sim} with a one-dimensional covariate,
we estimate the observational treatment-control contrast $\Delta(\cdot)$ by applying the DR-learner~\citep{kennedy2023towards} to the observational dataset, 
with generalized additive models (GAMs) as implemented by the \texttt{gam} function in the \texttt{mgcv} package in R~\citep{mgcv} as base learners and the observational propensity score
\[
r(x) = \Pr_{\obs}(Z = 1 \mid X=x)
\]
estimated using a GAM with logit link.
For all of our GAM fits,
we use the default behavior for a ``smooth" term \texttt{s()} in the \texttt{gam} function (thin-plate regression splines with degrees of freedom chosen automatically via restricted maximum likelihood).
For the AIPSW and collaborative estimators,
we estimate the CATE $\mu(\cdot)$ using the super learner~\citep{superlearner} with $L_1$-penalized linear regression~\citep{tibshirani1996regression}, multivariate adaptive regression splines~\citep{friedman1991multivariate}, random forests~\citep{breiman2001random}, and support vector machines~\citep{cortes1995support} as the component learners,
with the observations $\{D_i\}_{i=1}^n$ as the response.
We incorporate the prognostic adjustment idea of~\citet{liao2025prognostic} to improve learning of the CATE $\mu(\cdot)$ from the small experimental dataset by including both $X$ and our estimates $\hat{\Delta}(X)$ of the treatment contrast $\Delta(X)$ as predictors in the super learner.
The odds ratio $q_{n,N}$ for the AIPSW estimator is estimated by fitting a GAM with logit link to the combined experimental and observational dataset,
with the response variable being the binary indicator of whether the observation is in the experimental dataset.
The odds ratio $g_{n,N}$ for the collaborative estimator is estimated in the same way,
except with $\hat{\mu}(X_i)$ as the single predictor (instead of $X_i$).
Finally, the nuisance function $k_{n,N}(\cdot)$ for the collaborative estimator is estimated by fitting a GAM (with identity link) to predict $\hat{\mu}(X_i)$ from the estimates $\hat{g}_{n,N}(X_i)$.
Since $X$ is univariate, we do not use cross-fitting for any of our nuisance function estimates for the simulations in Section~\ref{sec:1d_sim}.
In particular, we set $K=1$ in the algorithm described at the start of Section~\ref{sec:estimation}.

For the simulations with ten covariates in Section~\ref{sec:multivar_sim},
we learn $\Delta(\cdot)$ using a causal forest,
as implemented by the \texttt{causal\_forest} function with default hyperparameters in the \texttt{grf} package in R~\citep{grf}.
To compute the estimates $\hat{\Delta}(X_i)$ in the observational dataset (observations $i=n+1,\ldots,n+N$),
we use the out-of-bag predictions in lieu of cross-fit predictions.
The odds ratio $q_{n,N}(\cdot)$ for the AIPSW estimator is estimated using logistic regression (with additive terms for all 10 covariates),
as advocated in the semi-supervised learning literature~\citep{zhang2023double}.
The remaining nuisance functions $\mu(\cdot)$, $g_{n,N}(\cdot)$, and $k_{n,N}(\cdot)$ for the AIPSW and collaborative estimators are estimated in the same way as for the univariate simulations in Section~\ref{sec:1d_sim},
except that for $\mu(\cdot)$ we now implement cross-fitting with $K=5$ folds due to the dimensionality of the covariate space relative to the experimental dataset size $n=100$.

Finally, for the simulations based on the crop rotation data in Section~\ref{sec:data},
we also estimate $\Delta(\cdot)$ using a causal forest (though with only 200 trees to speed up the computation time).
We estimate the odds ratio $q_{n,N}$ for the AIPSW estimator using $L_1$-penalized logistic regression with potential predictors including linear terms for all covariates as well as interactions up to order two between latitude and longitude.
The remaining nuisance functions are estimated in the same way as for the multivariate simulations, as described in the preceding paragraph.

\end{document}